\documentclass{article}
\usepackage{ijcai20}

\newif\ifdraft
\drafttrue

\usepackage[dvipsnames]{xcolor}

\usepackage{soul}
\usepackage{url}
\usepackage[hidelinks]{hyperref}
\usepackage[utf8]{inputenc}
\usepackage[small]{caption}
\usepackage{amsmath,amssymb,bm}
\usepackage{xspace}
\usepackage{booktabs}
\usepackage{latexsym}
\usepackage{float}
\usepackage{tikz}
\usepackage{framed}
\usepackage{multirow}
\usepackage{enumerate}
\usepackage[defblank]{paralist}
\usepackage{comment}
\usepackage{microtype}

\DeclareMathAlphabet{\altmathcal}{OMS}{cmsy}{m}{n}
\usepackage{newtxtext,newtxmath}

\usetikzlibrary{arrows,fit}
\usetikzlibrary{shapes}

\usepackage[thmmarks,amsmath,thref,hyperref]{ntheorem}

\theoremseparator{.}
\theorembodyfont{\itshape}
\newtheorem{theorem}{Theorem} 
\newtheorem{proposition}[theorem]{Proposition}
\newtheorem{lemma}[theorem]{Lemma}
\newtheorem{corollary}[theorem]{Corollary}

\theoremseparator{.}
\theorembodyfont{\upshape}
\theoremsymbol{\mbox{$\triangleleft$}}
\newtheorem{definition}{Definition}

\theorembodyfont{\upshape}
\theoremsymbol{\mbox{$\triangleleft$}}
\newtheorem{example}{Example}
\newtheorem{remark}{Remark}

\theoremstyle{nonumberplain}
\theoremheaderfont{\scshape}
\theorembodyfont{\upshape}
\theoremsymbol{\mbox{$\square$}} 
\newtheorem{proof}{Proof}

\makeatletter
\newcommand{\nobrackettag}[0]{\def\tagform@##1{\maketag@@@{##1}}}
\makeatother

\newcommand{\citeasnoun}[1]{\citeauthor{#1}~\shortcite{#1}}

\newcommand{\myparagraph}[1]{\smallskip\noindent\textbf{#1}~}

\newcommand{\A}{\mathcal{A}}
\newcommand{\B}{\mathcal{B}}
\newcommand{\C}{\mathcal{C}}

\newcommand{\I}{\mathcal{I}}

\newcommand{\F}{\mathcal{F}}
\newcommand{\G}{\mathcal{G}}
\renewcommand{\H}{\mathcal{H}}
\newcommand{\Q}{\mathcal{Q}}
\newcommand{\K}{\mathcal{K}}

\newcommand{\N}{\mathcal{N}}

\renewcommand{\P}{\mathcal{P}}
\newcommand{\R}{\mathcal{R}}

\newcommand{\T}{\mathcal{T}}

\newcommand{\tu}{\mathbf{t}}
\newcommand{\bc}{\mathbf{c}}

\newcommand{\s}{\mathbf{s}}
\renewcommand{\t}{\mathbf{t}}
\newcommand{\x}{\mathbf{x}}
\newcommand{\y}{\mathbf{y}}
\newcommand{\z}{\mathbf{z}}

\newcommand{\qq}{\mathfrak{q}}

\newcommand{\ex}[1]{\ensuremath{\mathsf{#1}}}
\newcommand{\exc}[1]{\ensuremath{\mathsf{#1}}}

\newcommand{\tup}[1]{\langle #1\rangle}            

\newcommand{\dom}[1][\I]{\Delta^{#1}}

\newcommand{\dllite}{\textit{DL-Lite}\xspace}
\newcommand{\dllitepos}{\textit{DL-Lite}\ensuremath{_{\mathit{pos}}}\xspace}
\newcommand{\dlliteposH}{\textit{DL-Lite}\ensuremath{_{\mathit{pos}}^{\H}}\xspace}
\newcommand{\dlliteposHmin}{\textit{DL-Lite}\ensuremath{_{\mathit{pos}}^{\H^-}}\xspace}
\newcommand{\dlliteposHminNmin}{\textit{DL-Lite}\ensuremath{_{\mathit{pos}}^{\H^{\text{--}}\!\N^{\text{--}}}}\xspace}
\newcommand{\dllitecore}{\textit{DL-Lite}\ensuremath{_{\mathit{core}}}\xspace}
\newcommand{\dllitecoreH}{\textit{DL-Lite}\ensuremath{_{\mathit{core}}^{\H}}\xspace}

\newcommand{\dllitecoreNmin}{\textit{DL-Lite}\ensuremath{_{\mathit{core}}^{\N^{\text{--}}}}\xspace}
\newcommand{\dllitecoreHNmin}{\textit{DL-Lite}\ensuremath{_{\mathit{core}}^{\H\N^{\text{--}}}}\xspace}
\newcommand{\dllitecoreHminNmin}{\textit{DL-Lite}\ensuremath{_{\mathit{core}}^{\H^{\text{--}}\!\N^{\text{--}}}}\xspace}

\newcommand{\dlliterb}{\textit{DL-Lite}\ensuremath{_{\R^-}^{\texttt{b}}}\xspace}

\newcommand{\dllitefb}{\textit{DL-Lite}\ensuremath{_{\F}^{\texttt{b}}}\xspace}

\newcommand{\owlql}{OWL\,2\,QL\xspace}

\newcommand{\LOGSPACE}{\textsc{LogSpace}\xspace}

\newcommand{\PTIME}{\textsc{P}\xspace}

\newcommand{\coNP}{co\textsc{NP}\xspace}

\newcommand{\ISA}{\sqsubseteq}
\newcommand{\GEQ}[2]{\mathop{\geq_{#1}}\!#2}

\newcommand{\NOT}{\lnot}

\DeclareMathOperator{\aggCount}{\mathsf{cnt}}

\DeclareMathOperator{\card}{\mathsf{card}}

\DeclareMathOperator{\domain}{\mathsf{dom}}
\DeclareMathOperator{\range}{\mathsf{range}}
\DeclareMathOperator{\inter}{\mathsf{inter}}
\DeclareMathOperator{\rdef}{\leftarrow}
\DeclareMathOperator{\ans}{\mathsf{ans}}

\DeclareMathOperator{\body}{\mathsf{body}}
\DeclareMathOperator{\head}{\mathsf{head}}

\DeclareMathOperator{\dist}{\mathsf{dist}}

\DeclareMathOperator{\matches}{\mathsf{match}}
\DeclareMathOperator{\certCard}{\mathsf{certCard}}
\DeclareMathOperator{\certAns}{\mathsf{certAns}}
\DeclareMathOperator{\partition}{\mathsf{part}}
\DeclareMathOperator{\atDec}{\mathsf{ad}}

\DeclareMathOperator{\qh}{\mathsf{qh}}
\DeclareMathOperator{\rpl}{\mathsf{rpl}}

\newcommand{\CQ}{$\mathsf{CQ}$\xspace}
\newcommand{\CQA}{$\mathsf{CQ^{\mathsf{A}}}$\xspace} 
\newcommand{\CQC}{$\mathsf{CQ^{\mathsf{C}}}$\xspace} 
\newcommand{\CQL}{$\mathsf{CQ^{\mathsf{L}}}$\xspace} 
\newcommand{\CQAC}{$\mathsf{CQ^{\mathsf{AC}}}$\xspace}
\newcommand{\CQAL}{$\mathsf{CQ^{\mathsf{AL}}}$\xspace} 
\newcommand{\CQCL}{$\mathsf{CQ^{\mathsf{CL}}}$\xspace} 
\newcommand{\AQ}{$\mathsf{AQ}$\xspace}
\newcommand{\CQR}{$\mathsf{CQ^{\mathsf{R}}}$\xspace} 
\newcommand{\CQCR}{$\mathsf{CQ^{\mathsf{CR}}}$\xspace} 
\newcommand{\CQCLR}{$\mathsf{CQ^{\mathsf{CLR}}}$\xspace} 

\newcommand{\can}{\I^\K_\mathit{can}}
\newcommand{\chase}{\textit{ch}}

\newcommand{\countP}{\textsc{Count}\xspace}

\newcommand{\targetLanguage}{FO(\textsc{Count})\xspace}
\newcommand{\comb}{\mathsf{cp}\xspace}
\newcommand{\setB}[2]{\left\{ #1 \mathrel{}\middle|\mathrel{} #2\right\}}

\newcommand{\sat}{\mathsf{subc}}

\newcommand{\perfectrefcount}{$\mathsf{PerfectRef}_{\mathsf{cnt}}$\xspace}
\newcommand{\perfectref}{$\mathsf{PerfectRef}$\xspace}
\newcommand{\atomrewrite}{\textit{AtomRewrite}\xspace}
\newcommand{\reduce}{\textit{Reduce}\xspace}
\newcommand{\gealpha}{\textit{GE$_\alpha$}\xspace}
\newcommand{\gebeta}{\textit{GE$_\beta$}\xspace}

\newcommand{\sql}{\textsf{SQL}\xspace}
\newcommand{\nr}[1]{\textcolor{blue}{#1}}
\newcommand{\ir}[1]{\textcolor{ForestGreen}{#1}}
\newcommand{\kr}[1]{#1}

\newcommand{\ignore}[1]{}

\newcommand{\NR}{\mathsf{N}_\mathsf{R}}
\newcommand{\NC}{\mathsf{N}_\mathsf{C}}
\newcommand{\NI}{\mathsf{N}_\mathsf{I}}
\newcommand{\NE}{\mathsf{N}_\mathsf{E}}
\newcommand{\NV}{\mathsf{N}_\mathsf{V}}

\newcommand{\ei}{\emph{(i)~}\xspace}
\newcommand{\eii}{\emph{(ii)~}\xspace}

\newcommand{\juliencolor}{red!70!black}

\newcommand{\diegocolor}{magenta!90!black}

\ifdraft
\newcounter{todocntr}
\newenvironment{todoenv}
  {\addtocounter{todocntr}{1}\par\bigskip\noindent
   \color{blue}\fbox{\thetodocntr}\hspace*{\fill}~\bf
   \begin{minipage}[t]{0.95\linewidth}}
  {\end{minipage}\hspace*{\fill}\bigskip}
\newcommand{\todo}[1]{\begin{todoenv}{#1}\end{todoenv}}
\newcommand{\td}[1]{\noindent{\color{red}\framebox[\linewidth]{\parbox{0.95\linewidth}{\bf #1}}}}
\newcommand{\boxcomment}[3]{~\\{\color{#2}{\fbox{\parbox{\linewidth}{{#1}:~#3}}}~\\}}
\makeatletter
\if@twocolumn
  \usepackage{marginnote}
  \reversemarginpar
  \newcommand{\nb}[2]{%
   \makebox[0cm][r]{\textcolor{#2}{\bf!}}%
   \marginnote
     {{\scriptsize\textcolor{#2}{#1}}}%
   }
  \newcommand{\inlinecomment}[3]{
   \marginnote[\textcolor{#2}{#1}]{\textcolor{#2}{#1}}\textcolor{#2}{#3}%
  }
\else
  \newcommand{\nb}[2]{\makebox[0cm][r]{\textcolor{#2}{\bf!}}%
   \marginpar[\parbox{\marginparwidth}{\scriptsize\textcolor{#2}{\raggedleft #1}}]%
   {\parbox{\marginparwidth}{\scriptsize\textcolor{#2}{\raggedright #1}}}}
  \newcommand{\inlinecomment}[3]{
   \marginpar[\textcolor{#2}{#1}]{\textcolor{#2}{#1}}\textcolor{#2}{#3}%
  }
\fi
\makeatother
\else
\newcommand{\todo}[1]{}
\newcommand{\td}[1]{}
\newcommand{\nb}[2]{}
\newcommand{\inlinecomment}[3]{}
\newcommand{\boxcomment}[3]{}
\fi

\newcommand{\dc}[1]{\nb{DC:~#1}{\diegocolor}}

\newcommand{\jc}[1]{\nb{JC:~#1}{\juliencolor}}

\begin{document}

\title{Counting Query Answers over a \dllite Knowledge Base (extended version)}

\author{
  Diego Calvanese$^{1,2}$ \and
  Julien Corman$^1$\and
  Davide Lanti$^1$ \textnormal{and}
  Simon Razniewski$^3$\\
  \affiliations
  $^1$\,Free University of Bozen-Bolzano, Italy\\
  $^2$\,Ume\aa\ University, Sweden\\
  $^3$\,Max-Planck-Institut f\"ur Informatik, Germany\\
  \emails
  \{calvanese, corman, lanti\}@inf.unibz.it,~ srazniew@mpi-inf.mpg.de
}
\maketitle

\ifdraft
\fi

\begin{abstract}
  Counting answers to a query
  is an operation
  supported by virtually all database management systems.
  In this paper we focus on counting answers over a Knowledge Base
  (KB), which may be viewed as a database enriched with background knowledge
  about the domain under consideration.
  In particular, we place our work in the context of Ontology-Mediated Query Answering/Ontology-based Data Access (OMQA/OBDA),
  where the language used for the ontology is a member of the \dllite family and the data is a (usually virtual) \emph{set} of assertions.
  We study the \emph{data complexity} of query answering,
  for different members of the \dllite family that include number restrictions,
  and for variants of \emph{conjunctive queries with counting} that differ with respect to their shape (connected,
  branching,
  rooted).
  We improve upon existing results by providing \PTIME and \coNP lower bounds,
  and upper bounds in \PTIME and \LOGSPACE.
  For the \LOGSPACE case, we have devised a novel query rewriting technique
  into first-order logic with counting.
\end{abstract}

\section{Introduction}
\label{sec:introduction}

Counting answers to a query is an essential operation in data management,
and is supported 
by virtually every database management system.
In this paper, we focus on counting answers over a Knowledge Base (KB),
which may be viewed as a database (DB) enriched with background knowledge about the
domain of interest.
In such a setting,
counting may take into account
two types of information:
grounded assertions (typically DB records),
and existentially quantified
statements (typically statistics).


As a toy example,
the following is an imaginary KB storing a parent/child relation,
where explicit instances (e.g., Alice is the child of Kendall) coexist with
existentially quantified ones (e.g., Parker has 3 children):
\[
  \begin{array}{l}
    \ex{hasChild}(\exc{Kendall},\exc{Alice})\\
    \ex{hasChild}(\exc{Jordan},\exc{Alice})\\
    \ex{hasChild}(\exc{Parker},\exc{Bob})\\
    \ex{hasChild}(\exc{Parker},\exc{Carol})\\
  \end{array}
  \qquad
  \begin{array}{l}
    \textit{"Kendall has 2 children"}\\
    \textit{"Parker has 3 children"}\\
    \textit{"A child has at most}\\
    \qquad\textit{2 parents"}
  \end{array}
\]
The presence of both types of information is common
when integrating multiple
data sources.  One source may provide detailed records (e.g., one record per
purchase, medical visit,
etc.),
whereas another source may only provide statistics (number of purchases, of
visits,
etc.), due to anonymization, access
restriction, or simply because the data is recorded in this way.

In such scenarios,
counting answers to a query over a KB may require operations that go beyond counting records.
E.g., in our example,
counting the minimal number of children that must exist according to the KB
(where children can be explicit or existentially quantified elements in the
range of \ex{hasChild}) requires some non-trivial reasoning.
The answer is 4:
Bob or Carol may be the second child of Kendall,
but Alice cannot be the third child of Parker (because Alice has two parents already),
so a fourth child must exist.

One of the most extensively studied frameworks for query answering over a KB is
Ontology Mediated Query Answering (OMQA)\footnote{Also referred to as OBDA (for
 Ontology Based Data Access),
 when emphasis in placed on mappings connecting external data sources
 to a TBox \cite{XCKL*18}.} \cite{CaDL08,BiOr15}. In OMQA, the
background knowledge takes the form of a set of logical statements, called the
\emph{TBox}, and the records are a set of facts, called the
\emph{ABox}.  TBoxes are in general expressed in
Description Logics (DLs), which are decidable fragments of First-Order
logic that typically can express the combination of explicit and
existentially quantified instances mentioned above.  Therefore OMQA
may provide valuable insight about the computational problem of
counting over such data (even though, in practice, DLs may not be the
most straightforward way to represent such data).

For Conjunctive Queries (CQs) and Unions of CQs (UCQs), DLs have been
identified with the remarkable property that query answering over a KB
does not induce extra computational cost (w.r.t.\ worst-case complexity),
when compared to query answering over a relational DB
\cite{XCKL*18}.  This key property has led to the development of
numerous techniques that leverage the mature technology of relational DBs to perform query answering over a
KB.  In particular, the \dllite family
\cite{CDLLR07,ACKZ09} has been widely studied and adopted in OMQA/OBDA systems,
resulting in the \owlql standard \cite{W3Crec-OWL2-Profiles}.

Yet the problem of \emph{counting} answers over a \dllite KB has seen relatively
little interest in the literature. In particular,
whether counting answers exhibits desirable computational properties analogous
to query answering is still a partly
open question for such DLs.
A key result for counting over \dllite KBs was provided by
\citeasnoun{KoRe15}, who also
formalized the semantics we adopt
in this paper (which we call \emph{count semantics}).  For
CQs interpreted under count semantics, they show
a \coNP lower bound in \emph{data complexity}, i.e., considering that
the sizes of the query and TBox are fixed.
However, their reduction relies on a CQ that
computes the cross-product of two relations,
which is unlikely to occur in practice.
Later on, it was shown\footnote{The result was stated for
  the related setting of \emph{bag semantics}, but the same
  reduction holds for count semantics as well.} by
\citeasnoun{NKKK*19} that \coNP-hardness still holds (for a more expressive
DL) using a branching and cyclic CQ without cross-product.
Building upon these results, we
further investigate how query shape affects tractability.

Another important question is whether relational DB technologies may
be leveraged for counting in OMQA, as done for boolean and enumeration
queries.  A key property here is \emph{rewritability}, extensively
studied for \dllite and UCQs \cite{CDLLR07}, i.e., the fact that a query over a KB may
be rewritten as an equivalent UCQ over its ABox only, intuitively
``compiling'' part of the TBox into this new UCQ.
An important
result in this direction was provided by \citeasnoun{NKKK*19}, but in the
context of query answering under \emph{bag semantics}.  For certain
\dllite variants, it is shown that queries that are \emph{rooted}
(i.e., with at least one constant or answer variable) can be rewritten
as queries over the ABox. Despite there being a correspondence between
bag semantics and count semantics, they show that these results do
not automatically carry over to query answering under count
semantics, due the way bag answers are computed in the presence of a
KB.


So in this work, we further investigate the boundaries of tractability and
rewritability for CQs with counting over a \dllite KB, with an
emphasis on DLs that can express statistics about missing information.
%
As is common for DBs,
we focus on \emph{data complexity},
i.e., computational cost in the size of the ABox (likely to grow much
faster than the query or TBox).

%
Due to space limitations, the techniques used to obtain our results are only
sketched, but full proofs are available in the extended version of this
paper \cite{CCLR-2020-arxiv}.

\section{Preliminaries and Problem Specification}
\label{sec:preliminaries}


We assume mutually disjoint sets $\NI$ of \emph{individuals} (a.k.a.\
\emph{constants}), $\NE$ of \emph{anonymous individuals} (induced by
existential quantification),
$\NV$ of \emph{variables},
$\NC$ of \emph{concept names} (i.e., unary predicates,
denoted with $A$),
and $\NR$ of \emph{role names} (i.e., binary predicates,
denoted with $P$).
%
In the following, boldface letters, e.g., $\bc$, denote tuples, and
we treat tuples as sets.

\paragraph{Functions, Atoms.}

$\domain(f)$ and $\range(f)$ denote the domain and range of a function
$f$.  Given $D \subseteq \domain(f)$, the function $f$ restricted to
the elements in $D$ is denoted $f|_D$.  A function $f$ is
\emph{constant-preserving} iff $c=f(c)$ for each
$c\in\domain(f)\cap\NI$.  If $S\subseteq\domain(f)$, we use $f(S)$ for
$\{f(s) \mid s \in S\}$.  If $\tu=\tup{t_1,\ldots,t_n}$ is is a tuple
with elements in $\domain(f)$, we use $f(\tu)$ for
$\tup{f(t_1), \ldots,f(t_n)}$.

An atom $a$ has
the form $A(s)$ or $P(s,t)$, with $A \in \NC$,
$P \in \NR$, and $s,t \in \NI \cup \NE \cup \NV$.

\paragraph{Interpretations, Homomorphisms.}

An \emph{interpretation} $\I$ is a FO structure
$\tup{\dom, \cdot^\I}$, where the \emph{domain} $\dom$ is a
non-empty subset of $\NI \cup \NE$, and the
\emph{interpretation function} $\cdot^\I$ is a function that maps each
constant $c \in \NI$ to itself (i.e., $c^\I = c$, in other words, we adopt
the \emph{standard names assumption}),
each concept name $A \in \NC$ to a set $A^\I \subseteq \dom$, and each role name
$P \in \NR$ to a binary relation $P^\I \subseteq \dom \times \dom$.

Given an interpretation $\I$ and a constant-preserving function $f$ with domain
$\dom$, we use $f(\I)$ to denote the interpretation defined by
$\dom[f(\I)] = f(\dom)$ and $E^{f(\I)} = f(E^{\I})$ for each
$E\in$ $\NC \cup \NR$.
Given two interpretations $\I_1$, $\I_2$, we use $\I_1 \subseteq \I_2$ as a
shortcut for $\dom[\I_1]\subseteq \dom[\I_2]$ and
$E^{\I_1} \subseteq E^{\I_2}$, for each $E \in \NC \cup \NR$.  A
\emph{homomorphism} $h$ from $\I_1$ to $\I_2$ is a constant-preserving function
with domain $\dom[\I_1]$ that verifies $h(\I_1) \subseteq \I_2$.  We note that
a set $S$ of atoms with arguments in $\NI \cup \NE$ uniquely
identifies an interpretation, which we denote with $\inter(S)$.

\paragraph{KBs, DLs, Models.}

\begin{figure}[tbp]
  \centering
  $R \longrightarrow P \mid P^- \qquad  B \longrightarrow A \mid \GEQ{1}{R}
  \qquad C \longrightarrow A \mid \GEQ{n}{R}$
  \caption{Syntax of \dllitecore \emph{roles} $R$, \emph{basic concepts} $B$,
   and \emph{concepts} $C$, where $n$ denotes a positive integer, i.e.,
   $n\in\mathbb{N}^+$.}
  \label{fig:dllite}
\end{figure}
A KB is a pair $\K=\tup{\T,\A}$, where $\A$, called \emph{ABox},
is a finite set
of atoms with arguments in $\NI$, and $\T$, called \emph{TBox}, is a finite set of \emph{axioms}.  We
consider DLs of the \dllite family~\cite{ACKZ09},
starting with the logic \dllitecore,
where each axiom has one
of the forms
\begin{inparaenum}[\itshape (i)]
\item $B\ISA C$ (concept inclusion),
\item $B\ISA\NOT C $ (concept disjointness), or
\item $R\ISA R'$ (role inclusion),
\end{inparaenum}
where now and in the following,
the symbols $R$, $B$, and $C$ are defined according to the grammar of
Figure~\ref{fig:dllite},
and are called respectively \emph{roles}, \emph{basic concepts}, and
\emph{concepts}.
Concepts of the form $\GEQ{n}{R}$ are called \emph{number restrictions}.
\dllitepos allows only for axioms of form~\textit{(i)}, with the
requirement that the number $n$ in number restrictions may only
be~1. In this work we study extensions to this logic along three
orthogonal directions:
\begin{inparaenum}[\itshape (1)]
\item allowing also for axioms of form~\textit{(ii)},
  indicated by replacing the subscript $_{\mathit{pos}}$ with $_{\mathit{core}}$;
\item allowing also for axioms of form~\textit{(iii)}, indicated by adding a
  superscript~$^{\H}$;
\item allowing for arbitrary numbers in number restrictions,
  but only on the right-hand-side (RHS) of concept inclusion, indicated by adding a
  superscript~$^{\N^\text{--}}$.
\end{inparaenum}
We also use the superscript $^{\H^-}$ for logics with role inclusions,
but with the restriction on TBoxes defined by \citeasnoun{NKKK*19},
which disallows in a TBox $\T$ axioms of the form $B \sqsubseteq\ \GEQ{n}{R_1}$ if $\T$ contains a role inclusion $R_1 \sqsubseteq R_2$,
for some $R_2 \neq R_1$.

The semantics of DL constructs
is specified as usual \cite{BCMNP03}.
An interpretation $\I$ is a \emph{model} of $\tup{\T,\A}$ iff
$\inter(\A)\subseteq\I$, and $E_1^\I \subseteq E_2^\I$ holds for each axiom
$E_1 \sqsubseteq E_2$ in $\T$.
A KB is \emph{satisfiable} iff it admits at least one model.  For readability,
in what follows we focus on satisfiable KBs, that is, we use ``a KB'' as a
shortcut for ``a satisfiable KB''.  We use the binary relation $\sqsubseteq_\T$
over \dllite concepts and roles $E_1$, $E_2$ to denote entailment w.r.t.\ a TBox
$\T$, defined by $E_1\sqsubseteq_\T E_2$ iff $E_1^\I \subseteq E_2^\I$ for
each model $\I$ of the KB $\tup{\T,\emptyset}$.

A key property of a \dllite KB $\K$ is the existence of a so-called
\emph{canonical model} $\can$, unique up to isomorphism, s.t.\ there exists a
homomorphism from $\can$ to each model of $\K$.  This model can be constructed
via the \emph{restricted chase} procedure by \citeasnoun{CDLLR13},
\citeasnoun{BoAC10}.

Finally, we observe that axioms of the form $B \sqsubseteq \GEQ{n}{R}$ can be
expressed in the logic \dllitecoreH, but with a possibly exponential blowup of
the TBox (assuming $n$ is encoded in binary).  For instance, the axiom
$A \sqsubseteq \GEQ{2}{P}$ can be expressed as
$\{ A \sqsubseteq \exists P_1, A \sqsubseteq \exists P_2, P_1 \sqsubseteq P,
P_2 \sqsubseteq P, \exists P^-_1 \sqsubseteq \neg \exists P^-_2 \}$, with
$P_1$, $P_2$ fresh DL roles.
  

\paragraph{CQs.}

A \emph{Conjunctive Query} (CQ) $q$ is an expression of the form
$q(\x) \rdef p_1(\tu_1), \ldots, p_n(\tu_n)$,
where each $p_i\in\NC\cup\NR$,
$\x\subseteq\NV$, each $\tu_i\subseteq\NV\cup\NI$, and $p_1(\tu_1), \ldots,
p_n(\tu_n)$ is syntactic sugar for the duplicate-free conjunction of atoms
$p_1(\tu_1) \land \cdots \land p_n(\tu_n)$.
Since all conjunctions in this work are duplicate-free, we sometimes treat
them as sets of atoms.
%
The variables in $\x$, called \emph{distinguished}, are denoted by
$\dist(q)$, $\head(q)$ denotes the \emph{head} $q(\x)$ of $q$, and
$\body(q)$ denotes the \emph{body}
$\{p_1(\tu_1), \ldots, p_n(\tu_n)\}$ of $q$.  We require safeness,
i.e., $\x \subseteq \tu_1 \cup\cdots\cup \tu_n$.  A query is
\emph{boolean} if $\x$ is the empty tuple.

\paragraph{Answers, Certain Answers.}

To define query answers under count semantics, we adapt the definitions by
\citeasnoun{CoNS07} and \citeasnoun{KoRe15}.
A \emph{match} for a query $q$ in an interpretation $\I$ is a
homomorphism from $\body(q)$ to $\I$.  Then, an \emph{answer} to $q$
over $\I$ is a pair $\tup{\omega,k}$ s.t.\ $k \ge 1$, and there are
\emph{exactly} $k$ matches $\rho_1,\ldots,\rho_k$ for $q$ in $\I$ that
verify $\omega = \rho_i|_{\dist(q)}$ for $i \in \{1,\ldots,k\}$.  We
use $\ans(q,\I)$ to denote the set of answers to $q$ over $\I$.
Similarly, if $\Q$ is a set of queries, we use $\ans(\Q,\I)$ to
denote the set of all pairs $\tup{\omega,\ell}$ s.t.\
$\tup{\omega,k} \in \ans(q,\I)$ for some $k$ and $q \in \Q$, and
$\ell= \sum \{k \mid \tup{\omega,k} \in \ans(q,\I), q \in \Q\}$.
Answering a query over an interpretation (i.e., a DB) is
also known as \emph{query evaluation}.  Finally, a pair
$\tup{\omega,k}$ is a \emph{certain answer} to a query $q$ over a KB
$\K$ iff $k \ge 1$ is the \emph{largest} integer such that, 
for each model $\I$ of $\K$, $\tup{\omega,k_{\I}} \in \ans(q,\I)$ for some $k_{\I}\ge k$.
We use $\certAns(q,\K)$ to denote the set of certain answers to $q$ over~$\K$.

\paragraph{Decision Problem.}

The decision problem defined by \citeasnoun{KoRe15} takes as input a query $q$,
a mapping $\omega$, a KB $\K$, and an integer $k$,
and decides $\tup{\omega, k} \in \certAns(q,\K)$.
It is easy to see that an instance of this problem can be reduced (in linear
time) to an instance where $q$ is a boolean query and $\omega$ is the empty
mapping, by introducing constants in $\body(q)$.
We will use the following simplified setting for the complexity results below:
if $q$ is a boolean query and $\varepsilon$ the empty mapping,
we use $k = \certCard(q,\K)$ as an abbreviation for $\tup{\varepsilon,k} \in
\certAns(q,\K)$. Then, the problem \countP is stated as follows:

\smallskip
\noindent
\framebox[\columnwidth]{\begin{tabular}{ll@{~}l@{~}}
   \countP
   & \textbf{Input}: & \dllite KB $\K$, boolean CQ $q$, $k \in \mathbb{N}^+$\\
   & \textbf{Decide}: & $k = \certCard(q, \K)$
 \end{tabular}}

\paragraph{Data complexity.}

As usual for query answering over DBs \cite{Vard82} or KBs \cite{CDLLR07},
we distinguish between \emph{combined} and \emph{data} complexity.
For the latter, we adopt the definition provided by \citeasnoun{NKKK*19},
i.e., we measure data complexity in the cumulated size of the ABox and the
input integer $k$ (encoded in binary).

\paragraph{Query Shape.}

As we will see later,
the shape of the input CQ may play a role for tractability.
We define here the different query shapes used throughout the article.
Because our focus is on queries with unary and binary atoms, we can use the Gaifman
graph \cite{BKKRZ17} of a CQ to characterize such shapes: the Gaifman graph
$\G$ of a CQ $q$ is the undirected graph whose vertices are the
variables appearing in $\body(q)$,
and that contains an edge between $x_1$ and $x_2$ iff
$P(x_1,x_2) \in \body(q)$ for some role
$P$.\footnote{%
 This definition implies that the Gaifman graph of $q$ has an edge from $x$ to
 $x$ if $P(x,x) \in \body(q)$.}  We call $q$ \emph{connected} (denoted with $q
\in$ \CQC) if $\G$
is connected, \emph{linear} ($q \in$ \CQL) if the degree of each vertex in
$\G$ is $\leq 2$, and \emph{acyclic} ($q \in$ \CQA) if $\G$ is acyclic. We note
that none of these three notions implies any of the other two.
In addition, following \citeasnoun{NKKK*19},
we call a CQ \emph{rooted} ($q \in$ \CQR) if each connected component in $\G$
contains at least one constant or one distinguished variable.
Finally, a CQ $q$ is \emph{atomic} ($q \in$ \AQ) if $|\body(q)| = 1$.

\paragraph{Rewritability.}
\label{sec:rewrit-def}

Given a query language $\Q$, a $\Q$-rewriting of a \CQ $q$ with respect to a KB
$\K=\tup{\T,\A}$ is a $\Q$ query $q'$ whose answers over $\inter(\A)$ alone
coincides with the certain answers to $q$ over $\K$.  For instance, for OMQA
with boolean or enumeration queries, $\Q$ is traditionally the language of
domain independent first-order 
queries, the logical underpinning of SQL.  As for queries with counting, it has
been shown by \citeasnoun{GrMi96}, \citeasnoun{NKKK*19} that counting
answers over a relational DB can be captured by query languages with evaluation
in \LOGSPACE (data complexity).

\section{Related Work}
\label{sec:sota}

Query answering under count semantics can be viewed as a specific case of query
answering under \emph{bag semantics}, investigated notably by
\citeasnoun{GrMi96} and \citeasnoun{LiWo97}, but for relational DBs rather than
KBs. Instead, in our setting, and in line with the OMQA/OBDA literature, we
assume that the input ABox is a set rather than a bag.
The counting problem over sets has also been studied recently in the
DB setting \cite{PiSk13,ChMe16}, but from the perspective of
combined complexity, where the shape of the query (e.g., bounded
treewidth) plays a prominent role.

As for (\dllite) KBs, \citeasnoun{CKNT08} define an alternative
(\emph{epistemic}) count semantics,
which counts over all grounded tuples (i.e., over $\NI$) entailed by the KB.
Such a semantics does not account for existentially implied individuals, and
thus cannot capture the statistics motivating our work.

\ignore{
 citazioni buttate NKKK*19,CNKK*19
 \[
   q() \rdef B(x_1), E(x_2, x_3), C(x_2, x_4), C(x_3, x_4)
 \]

 Concerning \countP, \cite{NKKK*19} introduced the notion of rewritability in
 terms of the language \textsc{BCALC}, which works on \emph{bags} (as opposed to
 sets from relational calculus), and thus naturally supports \textsc{count}
 queries.  Here we adopt this notion of rewritability.\jc{Shall we move the
  paragraph below to Section 1 or Section 3?}  A key distinction between the
 work by \cite{NKKK*19} and ours is that here we consider the ABox as a set of
 assertions, whereas they explore the possibility of having a bag of assertions.
 This choice is motivated by our focus on OBDA, where ABoxes are only
 \emph{virtual} conceptual views of the underlying physical data, and not actual
 assertions with a clearly-defined multiplicity.
}

Instead, the work closest to ours, and which first introduced the
count semantics that we adopt here, is the one of \citeasnoun{KoRe15}, who
first showed \coNP-hardness of the \textsc{Count} problem for data
complexity for \dllitepos, with a reduction that uses a disconnected
and cyclic query.  \coNP-membership is also shown for DLs up to
\dllitecoreH.

\citeasnoun{NKKK*19}, \citeasnoun{CNKK*19} have studied query answering over a
KB under
bag semantics, and provide a number of complexity results (including
\coNP-hardness) and query answering techniques (including a rewriting
algorithm). Such semantics is clearly related to the count semantics,
but there are notable differences as argued by \citeasnoun{NKKK*19}. In
short, one cannot apply the intuitive idea of treating sets as bags
with multiplicities $1$. Hence algorithms and complexity results
cannot be transferred between the two settings, and this already holds
for ontology languages that allow for existential restrictions on the
LHS of ontology axioms (note that \emph{all} the logics considered in
this paper allow for such construct).  The following example
by \citeasnoun{NKKK*19} illustrates this.

\begin{example} 
  Consider the KB $\K = \tup{\{A_1 \sqsubseteq \exists P,\, \exists P^-
   \sqsubseteq A\}$, $\{A_1(a),\, A_1(b)\}}$ and the query $q() \rdef A_2(y)$. If
  we evaluate this query in the count setting, then the answer is the
  empty tuple $\tup{}$ with cardinality $1$, because of the following model:
  
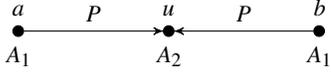
\begin{center}
    \begin{tikzpicture}[font=\footnotesize\sffamily, node distance=2cm]
      \node[circle, draw=black, fill=black, inner sep=.05cm,
        label={$a$}, label={below}:{$A_1$}] (a) at (0,0) {};
      \node[circle, draw=black, fill=black, inner sep=.05cm, right of=a,
        label={$u$}, label={below}:{$A_2$}] (u) {};
      \node[circle, draw=black, fill=black, inner sep=.05cm, right of=u,
        label={$b$}, label={below}:{$A_1$}] (b) {};

      \draw[-stealth'] (b) -- node[above] {$P$} (u);
      \draw[-stealth'] (a) -- node[above] {$P$} (u);
    \end{tikzpicture}
  \end{center}
  However, such structure does not accurately represent a bag interpretation. In
  fact, under bag semantics every concept and property is associated to a bag
  of elements. 
  Such bag can be seen as a function that returns, given an element, the number
  of times such element occurs in the bag.  We build now a (minimal)
  bag interpretation $\I$ for $\K$.  To satisfy $\A$, we set
  $A_1^\I(a) = 1$ and $A_1^\I(b) = 1$
  To satisfy $A_1 \sqsubseteq \exists P$, we introduce a single element $u$ (as
  above) and obtain $P^\I(a,u)=1$ and $P^\I(b,u)=1$. Therefore,
  $(P^-)^\I(u,a) = 1$ and $(P^-)^\I(u,b) = 1$, which, according to the
  semantics by \citeasnoun{NKKK*19}, imply that $(\exists P^{-})^{\I}(u)=2$.
  Therefore, to satisfy $\exists P^- \sqsubseteq A_2$, it must be that
  $A_2^\I(u) = 2$. In fact, the certain answer to $q$ over $\I$ under
  bag-semantics is the empty tuple $\tup{}$ with multiplicity $2$.
\end{example}



\section{Tractability and Intractability}
\label{sec:tractability}

We investigate now conditions for in/tractability (in data complexity) of
\countP, focusing on the impact of the shape of the query.
We observe that the
queries used by \citeasnoun{KoRe15} and \citeasnoun{NKKK*19} to show
\coNP-hardness are cyclic, and either disconnected or branching.
Building upon these results,
we further investigate
the impact of query shape on tractability.
We start with a membership result:
\begin{proposition}\label{prop:ptime-memb}
  \textsc{Count} is in \PTIME in data complexity for \dlliteposHminNmin and
  connected, linear CQs (\CQCL).
\end{proposition}

\begin{proof}[Sketch]
We start with \dlliteposHmin,
and then discuss how to extend the proof to \dlliteposHminNmin.
If $\K = \tup{\T,\A}$ is a \dlliteposHminNmin KB and $q$ a boolean query in
\CQCL, consider the set $\matches(q, \can)$ of all matches for $\body(q)$ over
the canonical model $\can$ of $\K$.  Then viewing $\matches(q, \can)$ as a
relation (i.e., a set of tuples), let $F_{\min}$ be the set of all
constant-preserving functions, whose domain is the set of all arguments in
$\matches(q, \can)$ and that minimize the number of resulting tuples when
applied to $\matches(q, \can)$.
%
Because $q$ is connected and linear,
and thanks to the limited expressivity of \dlliteposHminNmin,
it can be shown that there must be some $f \in F_{\min}$ that verifies $|f(\matches(q, \can))| = |\matches(q, f(\can))|$.
Since every model $\I$ of $\K$ verifies $\matches(q, \I) \subseteq
h(\matches(q, \can))$ for some homomorphism $h$,
and because $f(\can)$ is a model of $\K$,
it follows that $\certCard(q,\K) = |f(\matches(q, \can))|$.
Then it can also be shown that $|f(\matches(q, \can))|$ can be computed in time polynomial in $|\A|$.
  
  Now for \dlliteposHminNmin,
  to account for cardinality restrictions,
  we associate in every interpretation $\I$ a cardinality $\card_\I(e)$ to each $e \in \Delta^{\I}$:
  cardinality $1$ for elements of $\NI$,
  and possibly more than $1$ for elements of $\NE$.
  E.g., if $\K$ implies that an element $a \in \NI$ has 4 $P$-successors for
  some role $P$,
  and if there is only one $b \in \NI$ s.t.\ $(a,b) \in P^{\A}$,
  then $(a,e) \in P^{\can}$ for some $e \in \NE$,
  and $\card_{\can}(e) = 4-1 = 3$.
  Applying a function $f$ to an interpretation $\I$ affects these cardinalities:
  for each $e \in \Delta^{f(\I)}$,
  $\card_{f(\I)}(e) = \max\{\card_\I(e') \mid{} f(e') = e\}$.
  Then we extend cardinality to a tuple $\tu$ of elements, as
  $\card_\I(\tu) = \prod_{e \in \tu} \card_\I(e)$,
  and to a set $T$ of tuples, as
  $\card_\I(T) = \sum_{\tu \in T} \card_\I(\tu)$.
  In this extended setting,
  $\certCard(q,\K) = \card_{f(\I)}(\matches(q, \can))$
  for some function $f$ that minimizes $\card_{f(\I)}(\matches(q, \can))$.
  And this value can still be computed in time polynomial in $|\A|$.
\end{proof}

We now show that
for \dllitepos already,
disconnectedness alone leads to intractability,
i.e., cyclicity is not needed.
\begin{proposition}\label{prop:hardness_disconnected}
  \textsc{Count} is \coNP-hard in data complexity for \dllitepos and acyclic, linear, but disconnected CQs (\CQAL).
\end{proposition}
\begin{proof}[Sketch]
The proof is a direct adaptation of the one provided by \citeasnoun{KoRe15}.
We use a reduction from co-3-colorability to an instance of \textsc{Count}.
Let $\G = \tup{V, E}$ be an undirected graph with vertices $V$,
edges $E$, and without self-loops.
The ABox is $\A =\{\ex{Vertex}(v)\mid v \in V\}$ $\cup$
$\{\ex{edge}(v_1, v_2)\mid (v_1, v_2) \in E\}$ $\cup$
$\{\ex{Blue}(\ex{b}), \ex{Green}(\ex{g}), \ex{Red}(\ex{r}),$
$\{\ex{Color}(\ex{b}), \ex{Color}(\ex{g}), \ex{Color}(\ex{r}),$
$\ex{hasColor}(\ex{a}, \ex{b}), \ex{hasColor}(\ex{a}, \ex{g}),
\ex{hasColor}(\ex{a}, \ex{r}),$ $\ex{edge}(\ex{a},\ex{a})\}$
for some fresh constants $\ex{a}$, $\ex{b}$, $\ex{r}$, and $\ex{g}$.
The TBox is $\T = \{\ex{Vertex} \sqsubseteq \exists \ex{hasColor}, \exists
\ex{hasColor}^-  \sqsubseteq \ex{Color}\}$.
And the (acyclic, non-branching) query is
$q() \rdef$
$\ex{Color}(c),$
$\ex{edge}(v_1, v_2),$
$\ex{hasColor}(v_1, c_1),$
$\ex{hasColor}(v_2, c_2),$
$\ex{Blue}(c_1),$
$ \ex{Blue}(c_2),$
$\ex{edge}(v_3, v_4),$
$\ex{hasColor}(v_3, c_3),$
$\ex{hasColor}(v_4, c_4),$
$\ex{Green}(c_3),$
$\ex{Green}(c_4),$
$\ex{edge}(v_5, v_6),$
$\ex{hasColor}(v_5, c_5),$
$\ex{hasColor}(v_6, c_6),$
$\ex{Red}(c_5),$
$\ex{Red}(c_6)$.
Then it can be verified that $\certCard(q, \tup{\T,\A}) \geq 4$ iff $\G$ is not
3-colorable.
\end{proof}

Next we show that for the more expressive DL \dlliteposH,
branching alone leads to intractability:
\begin{proposition}\label{prop:hardness_branching}
  \textsc{Count} is \coNP-hard in data complexity for \dlliteposH and acyclic,
  connected, but branching CQs (\CQAC).
\end{proposition}




Finally, we observe that the \coNP upper bound provided by \citeasnoun{KoRe15}
for \dllitecoreH extends to \dllitecoreHNmin,\footnote{With a technicality:
 the input integer $k$ is not included in the notion of data complexity used
 by \citeasnoun{KoRe15}.} since number restrictions can be encoded in
\dllitecoreH,
as explained in Section~\ref{sec:preliminaries}.
\begin{proposition}
  \textsc{Count} is in \coNP in data complexity for \dllitecoreHNmin and arbitrary CQs (\CQ).
\end{proposition}


\section{Rewritability and Non-rewritability}
\label{sec:rewritability}


We now investigate conditions for rewritability.
We start by showing \PTIME-hardness for DLs with role inclusions and disjointness, and atomic queries.
\begin{proposition}\label{prop:ptimeHard_atomic}
  \textsc{Count} is \PTIME-hard in data complexity for $\dllitecoreH$ and atomic queries (\AQ).
\end{proposition}
\begin{proof}[Sketch]
  We show a \LOGSPACE reduction from the co-problem of evaluating a boolean circuit where all gates are NAND gates~\cite{GrHR91}
  to
  \countP.
We view such a circuit as an interpretation $\C$ 
whose domain consists of the circuit inputs and gates.
$T_I$, $F_I$, and $T_O$ are unary predicates interpreted in $\C$ as the
positive circuit inputs,
the negative circuit inputs and the (unique) target gate respectively.
$P$ is a binary predicate s.t.\ $(q, g) \in P^\C$ iff gate $g$ has input $q$
(where $q$ can be either a circuit input or another gate).

The TBox $\T$ is defined by $\T = \P \cup \T_1 \cup \T_2$,
where $\P = \{P_T \sqsubseteq P, P_F \sqsubseteq P \}$,
$\T_1 = \{F_I \sqsubseteq F, T_I \sqsubseteq T, T_O \sqsubseteq T, T \sqsubseteq \neg F \}$
and $\T_2 = \{T \sqsubseteq \exists P_F{^-}, F \sqsubseteq (\geq_2 P_T{^-}), \exists P_T \sqsubseteq T, \exists P_F \sqsubseteq F\}$.\footnote{
  The axiom $\geq_2 P_T{^-}$ can be encoded into $\dllitecoreH$,
  as explained in Section~\ref{sec:preliminaries}.
}
Intuitively,
the unary predicates $T$ and $F$ correspond to gates that evaluate to true and false respectively in the circuit,
and binary predicates $P_T$ and $P_F$ specialize $P$ to positive and negative inputs.
$\T_2$ encodes constraints pertaining to NAND gates: a positive gate must have at least one negative input,
and a negative gate must have two positive inputs.
Then $\T_1$ enforce that no gate can be both positive and negative,
and that the circuit inputs and the output gate have the desired truth values.

Finally,
as a technicality,
the ABox $\A$ is an extension of $\C$,
i.e., $\C \subseteq \inter(\A)$. 
The domain of $\A$ contains 3 additional individuals $\mathit{t_1}, \mathit{t_2}$ and $\mathit{f}$,
and it extends $P^\C$ with $\bigcup_{i \in (T_I)^\C} \{P(\mathit{f},i), P(\mathit{t_1},i)\}$,
$\bigcup_{i \in (F_I)^\C} \{P(\mathit{t_1},i), P(\mathit{t_2},i)\}$,
and $\{P(\mathit{t_1},\mathit{f}), P(\mathit{t_2},\mathit{f}), P(\mathit{f},\mathit{t_1}), P(\mathit{f},\mathit{t_2}), P(\mathit{t_2}, \mathit{t_1}),P(\mathit{t_1}, \mathit{t_2})\}$.

Then it can be verified that $\C$ is a valid circuit iff there exists a model $\I$ of $\tup{\T,\A}$ s.t.\
$|P^\I| = |P^\A|$.
Now let $q$ be the query $q() \rdef P(x_1, x_2)$. 
It follows that $\C$ is not a valid circuit iff $|P^\A|+ 1 = \certCard(q, \tup{\T,\A})$.
\end{proof}

Assuming \PTIME $\nsubseteq$ \LOGSPACE,
this implies that for such DLs,
even atomic queries cannot be rewritten into a query language whose evaluation is in \LOGSPACE,
which is sufficient to capture counting over relational databases.
Interestingly,
the reduction can be adapted so that it uses instead a query that is rooted, connected and linear (but not atomic).
\begin{proposition}\label{prop:ptimeHard_rooted}
  \textsc{Count} is \PTIME-hard in data complexity for $\dllitecoreH$ and rooted, connected, linear queries (\CQCLR).
\end{proposition}


We now focus on positive results, and
rewriting algorithms.

\subsection{Universal Model}

We follow the notion of \emph{universal model} proposed by \citeasnoun{NKKK*19}:
a model $\I$ of a KB $\K$ is \emph{universal} for a class of queries $\Q$ iff
$\ans(q,\I) = \certAns(q,\K)$ holds for every $q \in \Q$.
\citeasnoun{NKKK*19} and \citeasnoun{CNKK*19} investigated the existence of a
universal model for queries evaluated under bag semantics.
As we discussed in Section~\ref{sec:sota}, these results carry over to the setting of count semantics,
but only for ontology languages that disallow existential restriction on
the LHS of ontology axioms.
The existence of such model was proved over the class \CQR, for the DL-Lite$^b$
members up to
\dlliterb \cite{NKKK*19} and \dllitefb \cite{CNKK*19}, with some
syntactic restrictions. It was also shown that \CQR queries can be
rewritten into (\emph{BCALC}) queries to be evaluated over the (bag) input ABox.
Neither of these logics is able to encode numbers in the TBox though,
therefore they cannot capture statistical information about missing data.
And as discussed in the introduction, this information may be important in
some applications \cite{ChMe16}, and is one of the motivations behind our work.
Note also that both logics allow for existentials on the LHS of axioms,
and therefore these results do not carry over to count semantics.

Our first result shows the existence of a universal model for
\CQCR and \dllitecoreNmin, and queries evaluated under count semantics.
Precisely,
the canonical model $\can$
obtained via the \emph{restricted chase} from \citeasnoun{CDLLR13}
and \citeasnoun{BoAC10}
is a universal model.
From now on, we denote by $\chase_{i}(\K)$ the set of atoms obtained after
applying the $i$-th chase step over the KB $\K$, and by
$\chase_{\infty}(\K)$ the (possibly infinite) set of atoms obtained by an unbounded number of applications.

\begin{proposition} \dllitecoreNmin has a universal model w.r.t.\
  \countP over \CQCR queries.
\end{proposition}

\begin{proof}[Sketch]
    Consider a \CQCR query $q$ and a KB $\K$, and let
    $\matches(q, \can)$ denote the set of all matches for $\body(q)$ over the
    canonical model $\can := \chase_{\infty}(\K)$.
    Let $\I$ be a model of
    $\K$. Then there must exist a homomorphism $\tau$ from $\can$
    to $\I$.
    One immediately obtains that
    $\tau(\matches(q,\can)) \subseteq \matches(q,\tau(\can))$,
    and therefore \ei $|\tau(\matches(q,\can))| \le |\matches(q,\tau(\can))|$.
    Then relying on the
    fact that $q$ is rooted, that the chase is restricted, and
    that \dllitecoreNmin does not allow for role subsumption,
    it can proven  that \eii $|\matches(q,\can)| \le |\tau(\matches(q,\can))|$.
    So from \ei and \eii,
    $|\matches(q,\can)| \le |\matches(q,\tau(\can))|$ holds.
    Then,
    since $\tau$ is a homomorphism from $\can$ to $\I$,
    $\tau(\can) \subseteq I$ must hold,
    and therefore $|\matches(q,\tau(\can))| \le |\matches(q,\I)|$. We
    conclude that $|\matches(q,\can)| \le |\matches(q,\I)|$.
\end{proof}

\subsection{Rewriting for \dllitecoreNmin}

We introduce \perfectrefcount, a rewriting algorithm for $\dllitecoreNmin$
inspired by \perfectref \cite{CDLLR06}, and show its correctness. There is a
fundamental complication in our setting, of which we provide an
example. Consider a CQ $q$, a \dllitecoreNmin KB
$\K$, and a query $q'$ among those produced by \perfectref or any other
rewriting algorithm for CQs. Then, each match $\omega'$ for $q'$ in
$\inter(\A)$ can be extended to the anonymous individuals so as to form a
``complete'' match $\omega$ for $q$ in $\inter(\chase_\infty(\K))$ in a
certain number of ways (dictated by the axioms in the ontology). From now on, we call such number the \emph{anonymous
  contribution relative to $q'$}.
The following example shows that the anonymous contribution is related to the
number restrictions occurring in $\K$.

\begin{example}
  Consider the query $q(x) \rdef P(x,y)$, and the KB
  $\K = \tup{\{A \sqsubseteq \GEQ{3}{P}\},\, \{A(a)\}}$. Starting from $q$,
  \perfectref will produce, as part of the final rewriting, a query
  $q'(x) \rdef A(x)$.  Note that there is a single match
  $\mu = \{x \mapsto a\}$ for $q'$ over $\inter(\A)$, and that $\mu$ can be
  extended into exactly three matches for $q$ in $\inter(\chase_\infty(\K))$,
  by mapping variable $y$ into some anonymous individual.
\end{example}

To deal with the fact that
the anonymous contribution to a count is a non-fixed quantity that depends on
the axioms in the ontology, our algorithm is substantially different from
\perfectref and significantly more complicated.
It is also not related to the one by \citeasnoun{NKKK*19}, which is based on
\emph{tree-witness rewriting} \cite{KiKZ12} rather than on \perfectref, and
was not designed for
settings where the anonymous contribution
is a non-fixed quantity.

Given a CQ $q$ and a TBox $\T$, \perfectrefcount produces a \emph{set}
$\Q'$ of queries such that, for any ABox $\A$,
$\ans(\Q',\inter(\A)) = \certAns(q, \tup{\T,\A})$.  Each query in $\Q'$ comes
with a \emph{multiplicative factor} that captures the anonymous contribution of
each match for that query.  Queries in $\Q'$ are expressed in a \emph{target
 (query) language}, for convenience named \targetLanguage, which is a
substantial enrichment of the one introduced in
Section~\ref{sec:preliminaries}, but has a straightforward
translation into \sql.  Note that we use \targetLanguage only to express
the rewriting,
while user queries over the KB are still plain CQs.

Following \citeasnoun{CoNS07}, \targetLanguage allows one to explicitly specify
a subset of the non-distinguished variables, called \emph{aggregation
 variables}, for which we count the number of distinct mappings (so far, we
implicitly assumed all non-distinguished variables as aggregation variables).
It also allows for a \emph{multiplicative factor} to be applied after the
(count) aggregation operator,
a restricted use of disjunctions, equalities
between terms, atomic negation, and the use of \emph{nested aggregation}
in the form of a special $\exists^{=i}$ operator (which intuitively corresponds to a
nested aggregation plus a boolean condition requiring the result of the
aggregation to be equal to $i$).

A \emph{query} in \targetLanguage is a pair $\tup{Q(\x, \aggCount(\y) \cdot
 \alpha), \Pi}$, where variables $\x$ are called \emph{group-by variables},
variables $\y$ are called \emph{aggregation variables} (intuitively,
$\aggCount(\y)$ corresponds to the SQL construct \textsc{count distinct}), $\x
\cap \y = \emptyset$, $\alpha \in \mathbb{N}$ is a positive multiplicative
factor, and $\Pi$ is a set of \emph{rules} $\setB{q^k(\x : \y) \rdef \psi^k}{1
 \le k \le m}$.  The symbol ':' in the head\footnote{Head and body of a rule
 are defined as for CQs.} of each rule is to distinguish between group-by and
aggregation variables. Each $\psi^k$ in $\Pi$ is a conjunction
$\psi^k_{pos} \land \psi^k_{neg} \land \psi^k_{eq} \land \psi^k_{\exists}$ of
positive atoms ($\psi_{pos}^k$),
negated atoms ($\psi_{neg}^k$),
equalities between terms ($\psi_{eq}^k$),
and special atoms ($\psi_{\exists}^k$), which
we call \emph{$\exists$-atoms}, of the form  $\exists^{=i}_z.\ P(w,z)$, where
$i \in \mathbb{N}_0$, $w \in \x \cup \y$, and $z$ is a variable that occurs
only once in $q$.

A mapping $\rho$ is a match for $\psi^k$ in an interpretation $\I$ if:
\begin{compactitem}
\item $\rho(\psi^k_{pos}) \subseteq \I$;
\item $\rho$ satisfies all equalities in $\psi^k_{eq}$;
\item there is no $\rho' \supseteq \rho$ such that $\rho'(E(\z)) \in \I$, for some $\lnot E(\z)$ in $\psi^k_{neg}$;
\item for each $\exists^{=i}_y R(w,z)$ in $\psi_{\exists}$, there are
  \emph{exactly} $i$ mappings $\rho_1, \ldots, \rho_i$ such that, for $j \in
  \{1, \ldots, i\}$ we have that
  \begin{inparablank}
  \item $\rho_j \supseteq \rho$ and
  \item $\rho_j(R(w,z)) \subseteq \I$.
  \end{inparablank}
\end{compactitem}
A mapping $\rho$ is a match for $\Pi$ in an interpretation $\I$, if for some
$q(\x : \y) \rdef \psi$ in $\Pi$ it is a match for $\psi$ in $\I$.  A mapping
$\omega$ is an
\emph{answer} to $\tup{Q(\x, \aggCount(\y) \cdot \alpha), \Pi}$ over $\I$ with
cardinality $k \cdot \alpha$ iff there are \emph{exactly} $k$ mappings
$\eta_1, \ldots, \eta_k$ such that, for $i \in \{1,\ldots, k\}$:
\begin{compactitem}
\item $\omega = \eta_i|_{\x}$, and
\item $\eta_i$ can be extended to a match $\rho$ for $\Pi$ in $\I$ s.t.\ $\rho|_{\x \cup \y} = \eta_i$.
\end{compactitem}

Note that our semantics also captures the case when the operator $\aggCount()$ is
over an empty set of variables (in that case, the $k$ above would be equal to
$1$).
This technicality is necessary for the presentation of the algorithm.

We are now ready to introduce \perfectrefcount.  Consider a \emph{satisfiable}
knowledge base $\K = \tup{\T, \A}$, and a \CQCR query $q(\x) \rdef
\psi(\x,\y)$. \perfectrefcount takes as input $q$ and $\T$ and initializes the
result set $\Q$ as
\[
  \{\tup{Q(\x, \aggCount(\y) \cdot 1), \{q(\x :\y) \rdef \psi(\x,\y)\}}\}.
\]
Then the algorithm expands $\Q$ by applying the rules \atomrewrite, \reduce,
\gealpha, and \gebeta until saturation, with priority for \atomrewrite and
\reduce. The set $\Q'$ obtained at the end of this process does not
necessarily contain just queries (in the sense of our definition above), and
hence needs to be \emph{normalized} (see later).

To define the rules of the algorithm, we first need to introduce some notation.
In the following, $P^-(w,z)$ stands for $P(z,w)$.  Hence, also $R(w,z)$ when
$R = P^-$ stands for $P(z,w)$.  We use  '$\_$' to denote a fresh
variable introduced during the execution of the algorithm. For a basic concept
$B$, notation $\xi(B,w)$ stands for $B(w)$ if $B \in \NC$, or $R(w,\_)$, if
$B = \GEQ{1}{R}$. Given a set $\B$ of basic concepts, $\sat_\T(\B)$ is defined
as the set of basic concepts $\setB{B'}{B' \sqsubseteq_\T B, B \in \B}$.  If
$\phi$, $\psi$ are two conjunctions of atoms and $a$ is an atom in $\phi$, we
use $\phi[a/\psi]$ (resp., $\phi[a/\top]$) to denote the conjunction identical
to $\phi$, but where $a$ is replaced with $\psi$ (resp., $a$ is deleted from
$\phi$).  By extension, if $r$ is a rule, $r[a/\psi]$ denotes the rule
$\head(r) \rdef \body(r)[a/\psi]$. If $B$ is a basic concept and $R$ a role,
$\card_\T(B,R)$ denotes the maximal $n$ s.t.\
$B \sqsubseteq \GEQ{n}{R} \in \T$. A variable $x$ is \emph{bound in a rule $r$}
if it is a group-by variable, or if it occurs more than once in the set of
positive atoms of $r$. We say that $x$ is \emph{$\alpha$-blocked} if it is
bound, or if it occurs more than once in $\head(r)$, or if it occurs in some
$\exists$-atom in $\body(r)$. Finally, $x$ is \emph{$\beta$-blocked} if it is
bound, or if it occurs more than once in $\head(r)$, or if it occurs in some
atom of the form $\exists^{=i}_z R(w,z)$ with $i > 0$.

To ease the presentation, for the exposition of rules \gealpha and \gebeta we
will ignore details concerning the renaming of variables, and assume that
variables belong to the input query.%

\myparagraph{\atomrewrite ($\rightsquigarrow_{AR}$).} ${\{s\qq_1,
 \ldots,\qq_k\}\! \rightsquigarrow_{AR}\! \{\qq_1, \ldots,
 \qq_{k-1}, \qq_k'\}}$ if
%
\begin{compactitem}
\item $\qq_k= \tup{Q(\x, \aggCount(\y) \cdot \alpha), \Pi}$;
\item for some $r \in \Pi$, for some $E(\z) \in \body(r)$, either:
  \begin{compactitem}
  \item $E(\z)$ is of the form $A(z)$, and $B \sqsubseteq A \in \T$, or
  \item $E(\z)$ is of the form $R(w,z)$, $B \sqsubseteq \GEQ{n}{R} \in \T$, $z$
    is an unbound variable, and if $\head(r) = q(\s :\t)$, then $z \notin \t$;
\end{compactitem}
$\mathfrak{q'}_k = \tup{Q(\x, \aggCount(\y) \cdot \alpha),  \Pi \cup \{r[E(\z)/\xi(B,w)]\}}$.
\end{compactitem}
\myparagraph{\reduce ($\rightsquigarrow_{R}$).} $\{\qq_1, \ldots, \qq_k\}
\rightsquigarrow_{R} \{\qq_1, \ldots, \qq_{k-1}, \qq_k'\}$ if
\begin{compactitem}
\item $\qq_k=\tup{Q(\x, \aggCount(\y) \cdot \alpha), \Pi}$;
\item $\{E_1(\mathbf{\z_1})$, $E\mathbf{\z_2})\} \subseteq \body(r)$ for some $r \in \Pi$;
\item $\sigma$ is a most general unifier for $E(\z_1)$ and $E(\z_2)$;
  with the following restrictions:
  \begin{compactitem}
  \item a variable in $\x$ can map only to a variable in $\x$;
  \item a variable in $\y$ can map only to a variable in $\x \cup \y$;
  \item $\domain(\sigma) \subseteq \z_1 \cup \z_2$ and $\range(\sigma) \subseteq \z_1 \cup \z_2$.
  \end{compactitem}
\item $\qq'_k = \tup{Q(\x, \aggCount(\y) \cdot \alpha), \Pi \cup \{\sigma(r[E(\z_2)/\top])\}}$.
\end{compactitem}

\myparagraph{\gealpha ($\rightsquigarrow_{\ge_\alpha}$).} $\{\qq_1, \ldots, \qq_k\} \rightsquigarrow_{\ge_\alpha} \{\qq_1, \ldots, \qq_k\} \cup \Q_k$ if
\begin{compactitem}
\item $\qq_k = \tup{Q(\x, \aggCount(\y) \cdot \alpha), \Pi}$
\item[$\star$] $R(w,z)$, with $w,z \in \x \cup \y$, is an atom such that
  \[\Pi' \!:=\! \left\{\!
      R(w,z) \in \body(r)
      \middle\vert\!\!
      \begin{array}{l}
        r \in \Pi \text{ and } \\
        y \text{ is a non-$\alpha$--blocked}\\
        \text{aggregation variable}
      \end{array}
      \!\!\right\} \neq \emptyset
  \] 
\item Let $\psi_\exists$ be the conjunction of all exists-atoms in any rule
    $r \in \Pi$ (by construction, such conjunction is the same for all rules in $\Pi$).
    Then the conjunction $\psi_\exists \land \exists^0_y R(w,z)$ (seen as a
    set) must not appear in other rules from
    $\{\qq_1, \ldots, \qq_k\}$;
\item $\B_k$ is the maximal set of basic concepts $B$ such that $B \sqsubseteq \GEQ{n_B}{R} \in \T$,
  for some $n_B$;
\item [$\diamond$] $\Q_k$ is defined as follows.
  First, let $\partition(\B_k)$ denote the set of all pairs $\tup{\B^1,\B^2}$ such that
  $\B^1 \subseteq \B_k$,
  $\B^2 \subseteq \B_k$,
  $\B^1 \neq \emptyset$,
  and $\sat_\T(\B^1) \cap \sat_\T(\B^2) = \emptyset$.
  Then,
  for a set $\B$ of basic concepts, let $\comb_\T(\B)$ denote the cartesian
  product $\prod_{B \in \B} \sat_\T(B)$.
  And if $\B$, $\B'$ are two sets of basic concepts,
  we call \emph{atomic decomposition} the formula $\atDec(\B, \B')$, defined
  as:
  \[
    \textstyle
    \bigwedge_{B \in \B} \xi(B,w)
    \land \bigwedge_{B \in \B',B' \in\ \sat_\T(B)}
    \lnot \xi(B',w)
  \]
  If $\psi$ is a formula, let $\rpl(r,\psi)$ designate the rule:
  \[q(\s:\tu \setminus \{y\}) \rdef \body(r)[R(w,z)/\psi]\]
  Finally, if $j$ is an integer, let $\qh(j)$ be the expression:
  \[Q(\x, \aggCount(\y \setminus \{y\}) \cdot j \cdot \alpha)\]
  We can now define  $\Q_k$ as:\\
  \[\bigcup\limits_{\substack{(\B^1,\B^2) \in \partition(\B_k), \\
        n = \max_{B \in \B^1} \card(B),\\
        i \in [0 .. n-1]
      }}
    \begin{array}[t]{@{}l}
      \hspace{-9mm}\{\tup{\qh(n{-}i), \{\rpl(r, \atDec(\B^3, \B^2) \land \exists^{=i}_y R(w,z)) \\
      \hspace{1.2cm} \mid{} \B^3 \in \comb_\T(\B^1), r \in \Pi'\}}\}
    \end{array}\]
\end{compactitem}

\myparagraph{\gebeta ($\rightsquigarrow_{\ge_\beta}$).} This rule is
defined as $\rightsquigarrow_{\ge_\alpha}$, but with the difference
that conditions $\star$ and $\diamond$ are as follows:
\begin{compactitem}
\item[$\star$] $R(w,z)$ is an atom such that:
  \[\Pi' \!:=\! \left\{
      R(w,z) \in \body(r)
      \middle\vert
      \begin{array}{l}
        r \in \Pi \text{ and } \\
        y \text{ is a non-$\beta$--blocked}\\
        \text{aggregation variable}
      \end{array}
    \right\} \neq \emptyset
  \] 
\item[$\diamond$] As item $\diamond$ for $\gealpha$, with the additional condition
  that all atoms in which variable $y$ occurs are removed.
\end{compactitem}

Note that, once all rules have been applied to saturation,
the resulting set $\Q'$ is technically not yet a set of queries, because of renamed variables, constants, or repetitions in the head of a rule. To transform each element $\tup{Q(\x, \aggCount(\y) \cdot \alpha), \Pi}$ of $\Q'$ into a query, we \emph{normalize} it by renaming the variables in rules in $\Pi$, based on their positions,
according to $\x$ and $\y$, and by replacing constants and repeated variables in the head of a rule with suitable equalities in its body.

The intuition behind \perfectrefcount is the following.
First of all, we observe that the rewriting rules \atomrewrite and \reduce are analogous to their counterparts in the original \perfectref
algorithm. The restrictions on the unifier in \reduce are meant to limit the possible renamings of variables.
%
%
%
Rewriting rules \gealpha and \gebeta
extend the way existential quantification is handled in \perfectref,
and are the only ones eliminating aggregation variables from the rules in $\Pi$.
Each time one such variable is eliminated, it can be
potentially mapped in $(n{-}i)$ different ways into the anonymous
part of the canonical model. The $\exists$-atoms, together
with the relative atomic decompositions, check the number $i$ of
mappings that are already present in the ABox. The factor $\alpha$ keeps
track of the number of ways variables eliminated in previous steps can be
mapped into the anonymous part. Hence, the quantity
$(n{-}i) \cdot \alpha$ captures the anonymous contribution relative to the query.
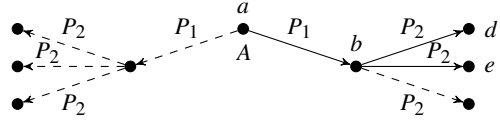
\begin{figure}
  \centering
  \tikzstyle{obj}=[circle, draw=black,fill=black,inner sep=.05cm]
  \begin{tikzpicture}[font=\footnotesize\sffamily]
    \node[obj, label={$a$}, label={below}:{$A$}] (a) at (0,0) {} ;

    \node[obj, label={$b$}] (b) at (1.5,-0.5) {} ;
    \node[obj] (c) at (-1.5,-0.5) {} ;

    \node[obj, label={right}:{$d$}] (d) at (3,0) {} ;
    \node[obj, label={right}:{$e$}] (e) at (3,-0.5) {} ;
    \node[obj] (f) at (3,-1) {} ;

    \node[obj] (g) at (-3,0) {} ;
    \node[obj] (h) at (-3,-0.5) {} ;
    \node[obj] (i) at (-3,-1) {} ;

    \draw[-stealth'] (a) -- node[above] {$P_1$} (b);
    \draw[-stealth',dashed] (a) -- node[above] {$P_1$} (c);

    \draw[-stealth'] (b) -- node[above] {$P_2$} (d);
    \draw[-stealth'] (b) -- node[above,shift={(3.5mm,-2pt)}] {$P_2$} (e);
    \draw[-stealth',dashed] (b) -- node[below] {$P_2$} (f);

    \draw[-stealth',dashed] (c) -- node[above] {$P_2$} (g);
    \draw[-stealth',dashed] (c) -- node[above,shift={(-3.5mm,-2pt)}] {$P_2$} (h);
    \draw[-stealth',dashed] (c) -- node[below] {$P_2$} (i);
  \end{tikzpicture}
  \caption{Chase model of Example~\ref{e:rewriting-ex}. Solid arrows represent
   the information in the ABox, whereas dashed lines represent information
   implied by the ontology.}
  \label{f:model-example}
\end{figure}
\begin{example}
  \label{e:rewriting-ex}
  Consider the KB $\K=\tup{\T,\A}$, with
  \[
    \T =
    \left\{\!\!\begin{array}{l@{\hspace{1ex}}l@{\hspace{1ex}}l}
        A & \sqsubseteq & \ge_2 P_1, \\
        \exists P_1^- & \sqsubseteq & \ge_3 P_2
      \end{array}\!\!
    \right\},
    \quad
    \A =\left\{\!\!
      \begin{array}{l@{\hspace{1ex}}l@{\hspace{1ex}}l}
        A(a),~ P_1(a,b), \\
        P_2(b,d),~ P_2(b,e)
      \end{array}\!\!
    \right\}
  \]
  and the input CQ $q(x) \rdef A(x), P_1(x,y_1), P_2(y_1,y_2)$.  The chase
  model of $\K$ is represented in Figure~\ref{f:model-example}. The
  initialization step sets
  $\Q = \{\tup{Q(x,\aggCount(y_1,y_2) \cdot 1)$, $\{q(x{:}y_1,y_2) \rdef A(x),
   P_1(x,y_1), P_2(y_1,y_2)\}}\}$. Since $y_2$ is unbound, we can apply rule
  $\gealpha$, which produces the following set $\Q'$:
  \[\left\{\hspace{-.2cm}\begin{array}{l}
        \tup{Q(x,\aggCount(y_1,y_2) \cdot 1), \\
         \quad \{q(x : y_1,y_2) \rdef A(x), P_1(x,y_1), P_2(y_1,y_2)\}},\\
        \tup{Q(x, \aggCount(y_1) \cdot 3 - 0), \\
         \quad \{q(x : y_1) \rdef A(x), P_1(x,y_1), P_1(\_,y_1), \exists^{=0}_z P_2(y_1,z)\}},\\
        \tup{Q(x, \aggCount(y_1) \cdot 3 - 1), \\
         \quad \{q(x : y_1) \rdef A(x), P_1(x,y_1), P_1(\_,y_1), \exists^{=1}_z P_2(y_1,z)\}},\\
        \tup{Q(x, \aggCount(y_1) \cdot 3 - 2), \\
         \quad \{q(x : y_1) \rdef A(x), P_1(x,y_1), P_1(\_,y_1), \exists^{=2}_z P_2(y_1,z)\}}\\
      \end{array}\hspace{-.2cm}\right\}
  \]
  Rule \reduce can now be triggered by the second, the third, and the last rule
  in $\Q'$.  In particular, an application of $\reduce$ on the second query
  leads to the query:
  \[\begin{array}{l}
      \langle Q(x, \aggCount(y_1) \cdot 3 - 0), \\
      \quad \{q(x : y_1) \rdef A(x), P_1(x,y_1), P_1(\_,y_1), \exists^{=0}_z P_2(y_1,z), \\
      \quad ~~q(x : y_1) \rdef A(x), P_1(x,y_1), \exists^{=0}_z P_2(y_1,z)\}\rangle\\
    \end{array}\]
  On such query we can apply rule $\gebeta$ producing, among others, the following query:
                \[\begin{array}{l}
                    \tup{Q(x, 1 \cdot (2-1) \cdot 3), \{q(x : \tup{}) \rdef A(x), \exists^{=1}_z P_1(x,z)\}}\\
                  \end{array}
                \]
                Let us analyze the queries produced by \perfectrefcount that return
                at least one answer. The query after the initialization step returns
                the number of paths $(x,y_1,y_2)$ in $\A$ conforming to the structure
                dictated by the body of the input query. Since there are two such
                paths, such query returns the answer $\tup{x \mapsto a, 2}$. The
                queries generated by $\gealpha$ check for all sub-paths $(x,y_1)$ of
                $(x,y_1,y_2)$ such that $x$ is an element of $A$, $y_1$ is a
                $P_1$-successor of $x$, and $y_1$ has less $P_2$-successors in the ABox
                than what the TBox prescribes.  There is one such path in $\can$,
                namely the one terminating in node $b$ that has only two
                $P_2$-successors in $\A$.  This path is captured by the fourth query
                in $\Q'$, which returns as answer $\tup{x \mapsto a, 1}$: indeed,
                there is a single way of extending this path into the anonymous part.
                The queries generated by $\gebeta$ are to be interpreted in a similar
                way. In particular, the query we highlighted retrieves the individual
                $a$, since this node has only one $P_1$-successor in $\A$ but it
                should have at least two $P_1$-successors according to $\T$. The
                answer to such query is $\tup{x \mapsto a, 3}$. Indeed, there are
                three ways of extending the match $x \mapsto a$ into the anonymous
                part.  Summing up the numbers, we get that our set of queries returns
                the answer $\tup{x \mapsto a, 6}$, which indeed is the answer to our
                input query over the chase model from Figure~\ref{f:model-example}.
              \end{example}

The algorithm terminates because the application of \atomrewrite
and \reduce is blocked upon reaching saturation, and each application of
\gealpha and \gebeta reduces the number of variables in $\psi$ by $1$.
The following lemmas show the correctness of \perfectrefcount.

\begin{lemma}
  \label{lemma:soundness}
  Consider a \dllitecoreNmin knowledge base $\K = \tup{\T, \A}$ and a
  connected, rooted CQ $q$. Consider a query
  $Q'$ belonging to the output of \perfectrefcount over $q$ and $\K$. Then, each match $\omega'$ for $Q'$ in
  $\inter(\A)$ can be extended into a match $\omega$ for $q$ in
  $\inter(\chase_{\infty}(\K))$.
\end{lemma}

\begin{proof}[Sketch]
The claim can be proved through a straightforward induction over the
number of chase steps.
\end{proof}

The next lemma states that the opposite direction also holds, i.e., that
\emph{all} matches are retrieved.

\begin{lemma}
  \label{lemma:completeness}
  Consider a \dllitecoreNmin knowledge base $\K = \tup{\T, \A}$ and a
  connected, rooted CQ $q$.
  Every match $\omega$ for $q$ in $\inter(\chase_{\infty}(\K))$ is an
  extension of some match $\omega'$ for $Q'$ in $\inter(\A)$, where $Q'$
  belongs to the output of \perfectrefcount.
\end{lemma}

\begin{proof}[Sketch]
The proof follows the one by \citeasnoun{CDLLR06}, however one has to pay
attention to the fact that here we deal with matches rather than with
assignments for the distinguished variables.  Another technical difference with
that proof is that in our case the nodes in the \emph{chase tree} are sets of
atoms
rather than single
atoms.
\end{proof}

The last lemma tells us that our way of capturing the anonymous
contribution is indeed correct.

\begin{lemma}
  \label{lemma:anonymous-contribution}
  Consider a \dllitecoreNmin knowledge base $\K = \tup{\T, \A}$ and a
  connected, rooted CQ $q$. Consider a query
  $\tup{Q'(\x, \aggCount(\y) \cdot \alpha), \Pi}$ belonging to the output of
  \perfectrefcount over $q$ and $\K$. Then, each match $\omega'$ for $Q'$ in
  $\inter(\A)$ can be extended into exactly $\alpha$ matches $\omega$ for $q$
  in $\inter(\chase_{\infty}(\K))$ with $\range(\omega \setminus \omega') \subseteq \NE$.
\end{lemma}

\begin{proof}[Sketch]
By induction on the number of applications of \gealpha and \gebeta.  It uses
Lemma~\ref{lemma:soundness} and the fact that variables are never eliminated by
\atomrewrite or \reduce. The atomic decomposition in \gealpha and \gebeta
guarantees that all combinations of number restrictions are considered.
\end{proof}

\begin{proposition}
  Consider a \dllitecoreNmin knowledge base $\K = \tup{\T, \A}$ and a connected, rooted CQ $q$.
  Let $\Q$ be the set of queries returned by a run of \perfectrefcount over $q$ and $\K$.
  Then:
  \[
    \textstyle
    \ans(\Q,\inter(\A)) = \certAns(q,\K)
  \]
\end{proposition}
\begin{proof}[Sketch]
  The claim follows from Lemmas~\ref{lemma:completeness}
  and~\ref{lemma:anonymous-contribution}, and by observing that the query $Q'$ in Lemma~\ref{lemma:anonymous-contribution}
  is unique due to $\exists^=i$ expressions, atomic decompositions, and the restrictions on \gealpha and \gebeta.
  Therefore, matches are not counted twice.
\end{proof}

The execution of \perfectrefcount does not depend on the ABox.  Considering
that the evaluation of \targetLanguage queries is in \LOGSPACE in data
complexity, this yields:
\begin{proposition}\label{prop:ptimeHard_rooted}
  \textsc{Count} is in \LOGSPACE in data complexity for $\dllitecoreNmin$ and
  rooted, connected CQs.
\end{proposition}

\section{Conclusion and Perspectives}
\label{sec:conclusion}

\begin{table}[tbp]
  \setlength{\tabcolsep}{1em}
  \centering
  \small
  \begin{tabular}{@{}l@{}c@{\hspace{1ex}}c@{\hspace{1ex}}c@{\hspace{1ex}}c@{\hspace{1ex}}c@{\hspace{1ex}}c@{\hspace{1ex}}c@{\hspace{1ex}}c@{}}
    \toprule
    & \emph{\AQ,\CQCL}
    & \emph{\CQAC}
    & \emph{\CQCLR,\CQCR}
    & \emph{\CQAL}
    & \emph{\CQ}
    \\
    \midrule
    \dllitepos & \ir{P} & \kr{\coNP} &  \ir{L} & \nr{\coNP-c} & \coNP-c \\
    \dlliteposH &\kr{\coNP} & \nr{\coNP-c}  & \kr{\coNP} & \ir{\coNP-c} &\coNP-c \\
    \dlliteposHminNmin & \nr{\PTIME} & \kr{\coNP} &  \coNP & \ir{\coNP-c}  & \coNP-c\\
    \dllitecore & \kr{\coNP} & \kr{\coNP} & \ir{L}& \ir{\coNP-c} &  \coNP-c\\
    \dllitecoreNmin & \coNP& \coNP& \nr{L} & \ir{\coNP-c}& \coNP-c\\
    \dllitecoreH& \nr{\PTIME-h}/\coNP& \ir{\coNP-c} & \nr{\PTIME-h}/\coNP & \ir{\coNP-c} &  \coNP-c\\
    \bottomrule
  \end{tabular}
  \caption{Summary of complexity results ('-h' stands for '-hard', and '-c' for
    '-complete'). New
    bounds proved here are in blue, bounds that directly follow in green,
    and already known bounds in black.}
  \label{tab:results}
\end{table}

Table~\ref{tab:results} summarizes our results for data complexity of query answering under count semantics for variants of CQs and \dllite.
%
Among other observations, these results indicate that for certain DLs,
whether a CQ is connected and branching affects tractability.
An interesting open question in this direction is whether the
$\PTIME$-membership result for \dlliteposHminNmin and \AQ/\CQCL is tight.
Indeed, the P-hardness result provided for \AQ holds for a more expressive DL
(namely \dllitecoreH),
which allows for disjointness and arbitrary interactions between role
subsumption and existential quantification.

A main contribution of this work is the query rewriting
technique provided in Section~\ref{sec:rewritability}.  It shows that
for connected and rooted CQs, and for variants of \dllite with neither disjointness
nor role subsumption, rewritability into a variant of SQL with
aggregates can be regained.  An interesting open question is whether
rewritability still holds for rooted queries and \dllitecoreHminNmin,
i.e., when allowing for restricted role subsumption.

Finally, it must be emphasized that this work is mostly theoretical,
and does not deliver a practical algorithm for query answering under
count semantics over \dllite KBs. In particular, the definition of
data complexity that we adopted does not take into account the
cardinality restrictions that may appear in the TBox.  This is
arguable: in scenarios where these restrictions may encode statistics,
it is reasonable to consider that these numbers ``grow'' with the size
of the data. The rewriting defined in Section~\ref{sec:rewritability}
may produce a query whose size is exponential in such numbers
(when they are encoded in binary).  Therefore a natural continuation of this
work is to investigate how arithmetic operations and nested aggregation can be
used to yield a rewriting whose size does not depend on the numbers that appear
in cardinality restrictions.


\section*{Acknowledgements}

This research has been partially supported
by the Wallenberg AI, Autonomous Systems and Software Program (WASP) funded by
the Knut and Alice Wallenberg Foundation,
by the Italian Basic Research (PRIN)
project HOPE, 
by the EU H2020 project INODE, 
grant agreement 863410,
by the CHIST-ERA project PACMEL,
and
by the project IDEE (FESR1133) through the European Regional Development
Fund (ERDF) Investment for Growth and Jobs Programme 2014-2020.

\bibliographystyle{named}
\bibliography{main-bib}

  \newpage
  \appendix
 \onecolumn
\section{Proofs}
\label{sec:proofs}
\subsection{Notation}
\paragraph{Matches for a subquery.}
If $\phi(\y)$ is a conjunction of atoms and $\I$ an interpretation,
we may use $\phi^\I$ below to denote the set of homomorphisms from $\phi$ to $\I$ (both viewed as sets of atoms).\\
For instance,
if
$$\phi_0(\y) = \ex{P_1}(y_1) \wedge \ex{P_2}(y_1, y_2)$$
and
$$\I = \{\ex{
  P_1(a),
  P_1(b),
  P_2(a,c),
  P_2(a,d)
  P_2(b,e)
  }\}$$,
  then
  \[\phi_0^\I = \left\{
      \begin{array}{lll}
        \{y_1 \mapsto \ex{a},& y_2 \mapsto \ex{c}\},\\
        \{y_1 \mapsto \ex{a},& y_2 \mapsto \ex{d}\},\\
        \{y_1 \mapsto \ex{b},& y_2 \mapsto \ex{e}\}
      \end{array} 
    \right\}\]
  In addition,
  if $\phi(\y)$ is a conjunction of atoms,
  $c_1, .., c_n \in \NI \cup \NV$ and $y_1, .., y_n \in \y$, 
  we use $\phi[y_1/c_1, ..,y_n/c_n]$ to designate the formula identical to $\phi$,
  but where each occurrence of $y_i$ is replaced by $c_i$,
  for $i \in [1 .. n]$.

  \noindent For instance, in the above example,
  $$\phi_0[y_1/\ex{a}] = \ex{P_1}(\ex{a}) \wedge \ex{P_2}(\ex{a}, y_2)$$
  and
  \[\phi_0[y_1/\ex{a}]^\I = \left\{
      \begin{array}{lll}
        \{y_2 \mapsto \ex{c}\},\\
        \{y_2 \mapsto \ex{d}\}\\
      \end{array} 
    \right\}\]
  
\paragraph{Empty mapping.}
We use $\varepsilon$ to denote the mapping with empty domain.

\subsection{Proposition~\ref{prop:hardness_disconnected}}
\textsc{Count} is \coNP-hard in data complexity for \dllitepos and acyclic, linear, but disconnected CQs (\CQAL).

  \begin{proof}
  We show that the reduction provided in the proof sketch above is correct.\\
  First, for readability, we partition the ABox and query used in this reduction.\\
  The ABox is defined as $\A = \A_1 \cup \A_2 \cup \A_3$,
  where:
    \begin{itemize}
    \item 
    $\A_1 = \{\ex{Vertex}(v)\mid v \in V\} \cup \{\ex{edge}(v_1, v_2)\mid (v_1, v_2) \in E\}$
  \item
    $\A_2 = \{\ex{Blue}(\ex{b}), \ex{Green}(\ex{g}), \ex{Red}(\ex{r})\}$,
  \item  
    $\A_3 = \{\ex{hasColor}(\ex{a}, \ex{b}), \ex{hasColor}(\ex{a}, \ex{g}), \ex{hasColor}(\ex{a}, \ex{r}),$ $\ex{edge}(\ex{a},\ex{a})\}$
    \end{itemize}
    \noindent For readability, we also reproduce the TBox:   
  $$\T = \{\ex{Vertex} \sqsubseteq \exists \ex{hasColor}, \exists \ex{hasColor}^-  \sqsubseteq \ex{Color}\}$$

    Then the query $q$ is defined as:
    $$q() \rdef \phi_0 \wedge \phi_b \wedge \phi_g \wedge \phi_r$$
    where:
    \begin{itemize}
    \item $\phi_0= \ex{Color}(c)$
    \item $\phi_b = \ex{edge}(v_1, v_2) \wedge \ex{hasColor}(v_1, c_1) \wedge \ex{hasColor}(v_2, c_2) \wedge \ex{Blue}(c_1) \wedge \ex{Blue}(c_2)$
    \item $\phi_g = \ex{edge}(v_3, v_4) \wedge \ex{hasColor}(v_3, c_3) \wedge \ex{hasColor}(v_4, c_4) \wedge \ex{Green}(c_3) \wedge \ex{Green}(c_4)$
    \item $\phi_r = \ex{edge}(v_5, v_6) \wedge \ex{hasColor}(v_5, c_5) \wedge \ex{hasColor}(v_6, c_6) \wedge \ex{Red}(c_5) \wedge \ex{Red}(c_6)$
    \end{itemize}
    Observe that the sets of variables occurring in $\phi_0, \phi_b,\phi_g$ and $\phi_r$ are pairwise disjoint.\\

  \noindent Observe that the KB $\K = \tup{\T,\A}$ is satisfiable.\\
  
   \noindent We now show that $\certCard(q, \K) \geq 4$ iff $\G$ is not 3-colorable.
    \begin{itemize}
   \item $(\Rightarrow)$
     By contraposition.\\
     Let $\G$ be 3-colorable.\\
     We show that there is a model $\I$ of $\K$ such that $\tup{\varepsilon, 3} \in \ans(q,\I)$.\\
     It follows that $\certCard(q,\K) < 4$.\\
     
     \noindent$\I$ is defined as follows:
     \begin{itemize}
     \item for $E \in \{\ex{Vertex, edge, Color, Blue, Green, Red}\}$,
       $E^\I = E^{\A}$,
     \item 
      each $v \in \ex{Vertex}^\I$ has a unique \ex{hasColor}-successor, among $\{$\ex{b, g, r}$\}$,
      and such that no adjacent vertices have the same successor;
      this is possible since $\G$ is 3-colorable (by assumption).
     \end{itemize}
     It can be verified that $\I$ is a model of $\tup{\T, \A}$.
     
     Then $\phi_b^\I = \{\{
     v_1 \mapsto \ex{a},
     v_2 \mapsto \ex{a},
     c_1 \mapsto \ex{b},
     c_2 \mapsto \ex{b}\}\}$,
    and similarly for $\phi_g^\I$ and $\phi_r^\I$ (with $\ex{g}$ and $\ex{r}$ respectively instead of $\ex{b}$).\\
    So $|\phi_b^{\I} \wedge \phi_g^{\I} \wedge \phi_r^{\I}|  = |\phi_b^{\I}| \cdot |\phi_g^{\I}| \cdot |\phi_r^{\I}| = 1 \cdot 1 \cdot 1 = 1$.\\
    Finally, $\phi_0^{\I} = \{\{c \mapsto \ex{b}\}, \{c \mapsto \ex{g}\}, \{c \mapsto \ex{r}\}\}$.\\
    So $|(\phi_0 \wedge \phi_b \wedge \phi_g \wedge \phi_r)^{\I}| = 3 \cdot 1 = 3$.\\
    Therefore $\ans(q,\I) = \{\tup{\varepsilon, 3}\}$.

  \item $(\Leftarrow)$
    Let $\G$ be non 3-colorable.\\
    We show below that for any model $\I$ of $\K$,
    there is a $k_\I \ge 4$ s.t. $= \{\tup{\varepsilon, k_\I}\} \in \ans(q,\I)$.\\
    It follows that $\certCard(q,\K) \geq 4$.

    Let $\I$ be a model of $\tup{\T, \A}$.\\
    Then $\{\{
     v_1 \mapsto \ex{a},
     v_2 \mapsto \ex{a},
     c_1 \mapsto \ex{b},
     c_2 \mapsto \ex{b}\}\} \in \phi_b^\I$,
    and similarly for $\phi_g^\I$ and $\phi_r^\I$ (with $\ex{g}$ and $\ex{r}$ respectively instead of $\ex{b}$).\\
    therefore $|\phi_b^\I| \geq 1$ (and similarly for $\phi_g^\I$ and $\phi_r^\I$).\\
    
    Now because $(\ex{Vertex} \sqsubseteq \exists \ex{hasColor}) \in \T$,
    each $v \in \ex{Vertex}^\I$ has a \ex{hasColor}-successor in $\Delta^I$.\\
    Then we have two cases:
    \begin{itemize}
    \item There is a $\ex{w} \in \ex{Vertex}^\I$ and a $\ex{c} \not\in \{\ex{b}, \ex{g},\ex{r}\}$ s.t. $\ex{hasColor}(\ex{w}, \ex{c}) \in \ex{hasColor}^\I$.\\
      Because $(\exists \ex{hasColor}^-  \sqsubseteq \ex{Color}) \in \T$,
      $\{\ex{Color}(\ex{c}), \ex{Color(b)},\ex{Color(g)},\ex{Color(r)}\} \subseteq \ex{Color}^\I$ must hold.\\
      So $\ex{Color}^I > 3$.\\
      Therefore $|\phi_0^\I| > 3$,
      and $\ans(q,\I) = \{\tup{\varepsilon,k_\I}\}$ for some $k_\I > 3 \cdot 1 \cdot 1 \cdot 1 = 3$.
    \item For each $v \in \ex{Vertex}^\I$, there is a $c \in \{\ex{b}, \ex{g},\ex{r}\}$ s.t. $\ex{hasColor}(v, c) \in \ex{hasColor}^I$.\\
      Because $\G$ is not 3-colorable,
       there must be $\ex{w}_1, \ex{w}_2 \in \ex{Vertex}^I$ and $\ex{c} \in \{\ex{b}, \ex{g},\ex{r}\}$ such that $\ex{edge}(\ex{w}_1, \ex{w}_2) \in \ex{edge}^\I$ and\\
       $\{\ex{hasColor}(\ex{w}_1, \ex{c}),\ex{hasColor}(\ex{w}_2, \ex{c})\} \subseteq \ex{hasColor}^\I$.\\
       So at least one of the following holds:
       \[\{
       v_1 \mapsto \ex{w}_1,
       v_1 \mapsto \ex{w}_2,
       c_1 \mapsto \ex{c},
       c_2 \mapsto \ex{c}
       \} \in \phi_b^\I\],
       \[\{
       v_3 \mapsto \ex{w}_1,
       v_4 \mapsto \ex{w}_2,
       c_3 \mapsto \ex{c},
       c_4 \mapsto \ex{c}
       \} \in \phi_g^\I\]
     or
       \[\{
       v_5 \mapsto \ex{w}_1,
       v_6 \mapsto \ex{w}_2,
       c_5 \mapsto \ex{c},
       c_6 \mapsto \ex{c}
       \} \in \phi_r^\I\]

     Therefore at least one of
     $|\phi_b^\I| \ge 2$,
     $|\phi_g^\I| \ge 2$ or
     $|\phi_r^\I| \ge 2$
     must hold.\\
     Since we know that 
     $|\phi_b^\I| \ge 1$,
     $|\phi_g^\I| \ge 1$, 
     $|\phi_r^\I| \ge 1$ and 
     $|\phi_0^\I| \geq 3$,
     we get:
     $$|\phi_0^\I| \cdot |\phi_b^\I| \cdot |\phi_g^\I| \cdot |\phi_r^\I| \ge 3 \cdot 2 \cdot 1 \cdot 1 = 6$$
      So there is  a $k_\I > 3$ s.t. $\ans(q,\I) = \{\tup{\varepsilon,k_\I}\}$.
    \end{itemize}
  \end{itemize}
\end{proof}

\subsection{Proposition~\ref{prop:hardness_branching}}

  \textsc{Count} is \coNP-hard in data complexity for \dlliteposH and acyclic,
  connected, but branching CQs (\CQAC).

 \begin{proof}
   Again, we adapt the reduction from (co-)3-colorability from~\cite{KoRe15},
so that the Gaifman graph of the query used in the reduction is acyclic,
 and has a single connected component.\\
 Let $G = \tup{V, E}$ be an undirected graph without self-loop (with vertices $V$ and edges $E$).\\
 We assume w.l.o.g. that $V$ is nonempty.\\
 
 We define the KB $\K = \tup{\T,\A}$ as follows.
    \begin{itemize}
    \item $\A = \A_1 \cup \A_2 \cup \A_3 \cup \A_4 \cup \A_5$, where:
    \begin{itemize}
     \item 
       $\A_1 = \{\ex{Vertex}(v)\mid v \in V\} \cup \{\ex{edge}(v_1, v_2)\mid (v_1, v_2) \in E\}$
     \item
       $\A_2= \bigcup\limits_{v \in V}\{\ex{conn}(v,\ex{a}), \ex{conn}(\ex{a},v)\}$
   \item
     $\A_3 = \{\ex{s}(v,c) \mid v \in V \text{ and } c \in \{\ex{b, g, r}\}\}$
   \item  
     $\A_4 = \{\ex{edge}(\ex{a, a}), \ex{conn}(\ex{a,a}), \ex{hasColor}(\ex{a, b}), \ex{hasColor}(\ex{a, g}), \ex{hasColor}(\ex{a, r})\}$
   \item
     $\A_5 = \{\ex{Blue}(\ex{b}), \ex{Green}(\ex{g}), \ex{Red}(\ex{r})\}$
     \end{itemize}
   \item $\T = \T_1 \cup \T_2$, where:
     \begin{itemize}
     \item 
       $\T_1 = \{\ex{Vertex} \sqsubseteq \exists \ex{hasColor}\}$
     \item $\T_2 = \{\ex{s} \sqsubseteq \ex{u}, \ex{hasColor} \sqsubseteq \ex{u}\}$
   \end{itemize}
 \end{itemize}
 \noindent Observe that the KB $\K = \tup{\T,\A}$ is satisfiable.\\
 
 \noindent We now define the query $q$ as: 
$$q() \rdef \phi_{\ex{Blue}} \wedge \phi_{\ex{Green}} \wedge \phi_{\ex{Red}}  \wedge \phi_{\ex{conn}}  \wedge \phi_{\ex{u}}$$
where:
    \begin{flalign*}
    \phi_{\ex{Blue}} &= \ex{edge}(v_1, v_2) \wedge \ex{hasColor}(v_1, c_1) \wedge \ex{hasColor}(v_2, c_2) \wedge \ex{Blue}(c_1) \wedge \ex{Blue}(c_2)\\
    \phi_{\ex{Green}} &= \ex{edge}(v_3, v_4) \wedge \ex{hasColor}(v_3, c_3) \wedge \ex{hasColor}(v_4, c_4) \wedge \ex{Green}(c_3) \wedge \ex{Green}(c_4)\\
    \phi_{\ex{Red}} &= \ex{edge}(v_5, v_6) \wedge \ex{hasColor}(v_5, c_6) \wedge \ex{hasColor}(v_6, c_6) \wedge \ex{Red}(c_5) \wedge \ex{Red}(c_6)\\
    \phi_{\ex{conn}} &= \ex{conn}(v_1, v_3) \wedge \ex{conn}(v_3, v_5)\\
    \phi_{\ex{u}} &= \ex{s}(v_7, c_1) \wedge \ex{u}(v_7, c_7) 
    \end{flalign*}
  Observe that the Gaifman graph of $q$ is a tree.\\

  \noindent For readability, we will use $\phi'$ below as a shortcut to denote (the conjunction of) the first 4 subformulas of $q$,
  i.e.:
  $$\phi' \doteq \phi_{\ex{Blue}} \wedge \phi_{\ex{Green}} \wedge \phi_{\ex{Red}}  \wedge \phi_{\ex{conn}}$$
  so that
  $q$ can alternatively be defined as:
  $$q() \rdef \phi' \wedge \phi_{\ex{u}}$$
  Observe that $c_1$ is the only variable shared by $\phi'$ and $\phi_{\ex{u}}$.\\
  
  \noindent We will also use $\omega_{\ex{a}}$ to denote the following soluton mapping:
     \[
     \omega_{\ex{a}} \doteq 
       \left\{
       \begin{array}{llll}
     v_1 \mapsto \ex{a},& v_2 \mapsto \ex{a},& c_1 \mapsto \ex{b},& c_2 \mapsto \ex{b},\\
     v_3 \mapsto \ex{a},& v_4 \mapsto \ex{a},& c_3 \mapsto \ex{g},& c_4 \mapsto \ex{g},\\
     v_5 \mapsto \ex{a},& v_6 \mapsto \ex{a},& c_5 \mapsto \ex{r},& c_6 \mapsto \ex{r}
       \end{array}
     \right\}
   \]

  We now show that $\certCard(q,\K) \geq 3 \cdot |V| + 1$ iff $G$ is not 3-colorable.\\
   
   \begin{itemize}
    \item $(\Rightarrow)$
      By contraposition.\\
      Let $G$ be 3-colorable.\\
     We show that there is a model $\I$ of $\K$ such that $\tup{\varepsilon, 3 \cdot |V|} \in \ans(q,\I)$.\\
      It follows that $\certCard(q,\K) < 3 \cdot |V| + 1$.\\
      
    \noindent $\I$ is defined as follows:
      \begin{itemize}
      \item for $E \in \{\ex{Vertex, edge, hasColor, s, conn, Blue, Green, Red}\}$,
      $E^\I = E^{\A}$,
     \item $\ex{u}^\I = \ex{hasColor}^\I \cup \ex{s}^\I$,
    \item 
      each $v \in (\ex{Vertex}^\I)$ has a unique \ex{hasColor}-successor,
      among $\ex{\{b, g, r\}}$,
       such that no adjacent vertices have the same successor;
       this is possible since $G$ is 3-colorable (by assumption).
      \end{itemize}
    It can be verified that $\I$ is a model of $\K$.\\

    Then
    $\phi_{\ex{Blue}}^\I = \{\{
     v_1 \mapsto \ex{a},
     v_2 \mapsto \ex{a},
     c_1 \mapsto \ex{b},
     c_2 \mapsto \ex{b}
     \}\}$,
    and similarly for $\phi_{\ex{Green}}^\I$ and $\phi_{\ex{Red}}^\I$ (with $\ex{g}$ and $\ex{r}$ respectively instead of $\ex{b}$).\\
      Observe also that
    $\{
     v_1 \mapsto \ex{a},
     v_3 \mapsto \ex{a},
     v_5 \mapsto \ex{a}
     \}
     \in \phi_{\ex{conn}}^\I$.\\
     It follows that $(\phi')^\I = \{\omega_{\ex{a}}\}$.\\
In addition,
$\phi_{\ex{u}}[c_1/\ex{b}]^{\I} =  \{ \{v_7 \mapsto v, c_7 \mapsto c\} \mid v \in V, c \in \{\ex{b,g,r}\}\}$.\\
Then because $c_1$ is the only variable shared by $\phi'$ and $\phi_{\ex{u}}$,
and because $\omega_{\ex{a}}(c_1) = \ex{b}$,
we get:
\begin{flalign*}
(\phi' \wedge \phi_{\ex{u}})^\I & = \{\omega_{\ex{a}} \cup \omega' \mid \omega' \in  \phi_{\ex{u}}[c_1/\ex{b}]^\I\}\\
|(\phi' \wedge \phi_{\ex{u}})^\I| & =  |\phi_{\ex{u}}[c_1/\ex{b}]^\I|\\
|(\phi' \wedge \phi_{\ex{u}})^\I| & =  3 \cdot |V|
\end{flalign*}
So $\tup{\varepsilon, 3 \cdot |V|} \in \ans(q,\I)$.

   \item $(\Leftarrow)$
     Let $G$ be non 3-colorable.\\
     We show below that for any model $\I$ of $\K$, $|(\phi' \wedge \phi_{\ex{u}})^\I| \geq 3 \cdot |V| + 1$.\\
    It follows that $\certCard(q,\K) \geq 3 \cdot |V| + 1$.\\

     Let $\I$ be a model of $\K$.\\
     Then
    $\{
     v_1 \mapsto \ex{a},
     v_2 \mapsto \ex{a},
     c_1 \mapsto \ex{b},
     c_2 \mapsto \ex{b}
     \}\in \phi_{\ex{Blue}}^\I$,
    and similarly for $\phi_{\ex{Green}}^\I$ and $\phi_{\ex{Red}}^\I$ (with $\ex{g}$ and $\ex{r}$ respectively instead of $\ex{b}$).\\
      In addition,
    $\{
     v_1 \mapsto \ex{a},
     v_3 \mapsto \ex{a},
     v_5 \mapsto \ex{a}
     \}
     \in \phi_{\ex{conn}}^\I$.\\
It follows that  $\omega_{\ex{a}} \in (\phi')^\I$.\\
     
Now let $\Omega_{\ex{u}} = \{ \{v_7 \mapsto v, c_7 \mapsto c\} \mid v \in V, c \in \{\ex{b,g,r}\}\}$.\\
Then $|\Omega_{\ex{u}}| = 3 \cdot |V|$.\\
     In addition, because $\A_3 \subseteq \I$ and $(\ex{s} \sqsubseteq \ex{u}) \in \T$,
     it must be the case that $\Omega_{\ex{u}} \subseteq \phi_{\ex{u}}[c_1/\ex{b}]^{\I}$.\\
     

     Finally,
     because $(\ex{Vertex} \sqsubseteq \exists \ex{hasColor}) \in \T$,
     each $v \in V^\I$ has a \ex{hasColor}-successor in $\Delta^I$.\\

     Then we have two cases:
     \begin{itemize}
    \item There is a $w \in \ex{Vertex}^\I$ and a $c \not\in \{\ex{b}, \ex{g},\ex{r}\}$ s.t. $\ex{hasColor}(w, c) \in \ex{hasColor}^\I$.\\
      Then $\omega = \{v_7 \mapsto w, c_7 \mapsto c\} \in \phi_{\ex{u}}[c_1/\ex{b}]^\I$.\\
      And since $\Omega_{\ex{u}} \subseteq \phi_{\ex{u}}[c_1/\ex{b}]^\I$ also holds,
      $\Omega_{\ex{u}} \cup \{\omega\} \subseteq \phi_{\ex{u}}[c_1/\ex{b}]^\I$.\\
      Observe also that $\omega \not\in \Omega_{\ex{u}}$.\\
      Therefore:
      \begin{flalign}
        |\phi_{\ex{u}}[c_1/\ex{b}]^\I|& \geq |\Omega_{\ex{u}} \cup \{\omega\}|\\
        |\phi_{\ex{u}}[c_1/\ex{b}]^\I|& \geq 3 \cdot |V| + 1\label{eq:a31}
      \end{flalign}
      Now recall that $\omega_{\ex{a}} \in (\phi')^\I$, that $\omega_{\ex{a}}(c_1) = \ex{b}$,
      and that $c_1$ is the only variable shared by $\phi'$ and $\phi_{\ex{u}}$.\\
      It follows that:
      \begin{flalign}
      (\phi' \wedge \phi_{\ex{u}})^\I &\supseteq \{\omega_{\ex{a}} \cup \omega' \mid \omega' \in \phi_{\ex{u}}[c_1/\ex{b}]^\I\}\\
      |(\phi' \wedge \phi_{\ex{u}})^\I| &\geq |\{\omega_{\ex{a}} \cup \omega' \mid \omega' \in \phi_{\ex{u}}[c_1/\ex{b}]^\I\}|\\
      |(\phi' \wedge \phi_{\ex{u}})^\I| &\geq |\phi_{\ex{u}}[c_1/\ex{b}]^\I|\label{eq:a32}
    \end{flalign}
    So from~\ref{eq:a31} and~\ref{eq:a32}:
      \begin{flalign*}
      |(\phi' \wedge \phi_{\ex{u}})^\I| &\geq 3 \cdot |V| + 1
    \end{flalign*}

    \item For each $v \in \ex{Vertex}^\I$, there is a $c \in \{\ex{b}, \ex{g},\ex{r}\}$ s.t. $\ex{hasColor}(v, c) \in \ex{hasColor}^I$.\\
      Because $G$ is not 3-colorable,
       there must be $w_1, w_2 \in \ex{Vertex}^I$ and $c \in \{\ex{b}, \ex{g},\ex{r}\}$ such that $\ex{edge}(w_1, w_2) \in \ex{edge}^\I$ and\\
       $\{\ex{hasColor}(w_1, c),\ex{hasColor}(w_2, c)\} \subseteq \ex{hasColor}^\I$.\\

       W.l.o.g., let us assume that $c = \ex{b}$ (the proofs for $c = \ex{g}$ and $c = \ex{g}$ are identical).\\
       Then $\{
       v_1 \mapsto w_1,
       v_2 \mapsto w_2,
       c_1 \mapsto \ex{b},
       c_2 \mapsto \ex{b}\} \in \phi_{\ex{Blue}}^\I$.\\
       
    In addition:
    \begin{flalign*}
    &\{
     v_1 \mapsto \ex{a},
     v_2 \mapsto \ex{a},
     c_1 \mapsto \ex{g},
     c_2 \mapsto \ex{g}
     \} \in \phi_{\ex{Green}}^\I\\
    &\{
     v_1 \mapsto \ex{a},
     v_2 \mapsto \ex{a},
     c_1 \mapsto \ex{r},
     c_2 \mapsto \ex{r}
     \} \in \phi_{\ex{Red}}^\I\\
     &\{
     v_1 \mapsto w_1,
     v_3 \mapsto \ex{a},
     v_5 \mapsto \ex{a}
     \}
     \in \phi_{\ex{conn}}^\I
    \end{flalign*}
     It follows that $\omega_{w_1} \in (\phi')^\I$,
     where $\omega_{w_1}$ is the mapping defined by:
     \[
     \omega_{w_1} \doteq 
       \left\{
       \begin{array}{llll}
     v_1 \mapsto w_1,& v_2 \mapsto w_2,& c_1 \mapsto \ex{b},& c_2 \mapsto \ex{b},\\
     v_3 \mapsto \ex{a},& v_4 \mapsto \ex{a},& c_3 \mapsto \ex{g},& c_4 \mapsto \ex{g},\\
     v_5 \mapsto \ex{a},& v_6 \mapsto \ex{a},& c_5 \mapsto \ex{r},& c_6 \mapsto \ex{r}
       \end{array}
     \right\}
   \]

   Recall also that $\omega_{\ex{a}} \in (\phi')^\I$.\\
      Because $c_1$ is the only variable shared by $\phi'$ and $\phi_{\ex{u}}$,
      and because $\omega_{w_1}(c_1) = \omega_{\ex{a}}(c_1) = \ex{b}$,
      we have:
      \begin{flalign}
        (\phi' \wedge \phi_{\ex{u}})^\I &\supseteq \{\omega_{w_1} \cup \omega' \mid \omega' \in \phi_{\ex{u}}[c_1/\ex{b}]^\I\} \cup \{\omega_{\ex{a}} \cup \omega' \mid \omega' \in \phi_{\ex{u}}[c_1/\ex{b}]^\I\}\label{eq:34}
      \end{flalign}
      Now observe that $\domain(\omega_{w_1}) =  \domain(\omega_{\ex{a}})$ and $\omega_{w_1} \neq \omega_{\ex{a}}$.\\
      Therefore for any set $\Omega'$ of mappings,
      $\{\omega_{w_1} \cup \omega' \mid \omega' \in \Omega'\} \cap \{\omega_{\ex{a}} \cup \omega' \mid \omega' \in \Omega'\} = \emptyset$.\\
      So from~\ref{eq:34}, we get:
      \begin{flalign}
      |(\phi' \wedge \phi_{\ex{u}})^\I| &\geq 2 \cdot |\phi_{\ex{u}}[c_1/\ex{b}]^\I|\label{eq:35}
    \end{flalign}
   Finally, recall that $\Omega_{\ex{u}} \subseteq \phi_{\ex{u}}[c_1/\ex{b}]^\I$ and $|\Omega_{\ex{u}}| = 3 \cdot |V|$.\\
      So from~\ref{eq:35}, we get:
      \begin{flalign*}
      |(\phi' \wedge \phi_{\ex{u}})^\I| &\geq 2 \cdot |\Omega_{\ex{u}}|\\
      |(\phi' \wedge \phi_{\ex{u}})^\I| &\geq 6 \cdot |V|\\
      |(\phi' \wedge \phi_{\ex{u}})^\I| &\geq 3 \cdot |V| + 1
    \end{flalign*}
     \end{itemize}
     \end{itemize}
    \end{proof}


\end{document}